\documentclass[10pt]{article}
\usepackage{amsfonts}
\usepackage{amssymb}
\usepackage{amsthm}
\usepackage{amsmath}
\usepackage{authblk} 
\usepackage{graphicx}
\usepackage{hyperref}
\usepackage{enumerate}
\usepackage{tikz}
\usepackage{color}
\usepackage{eucal}

\usepackage{soul,todonotes}
\usepackage{mathrsfs}
\usepackage{hyperref}
\usepackage{datetime}

\DeclareMathAlphabet{\mathpzc}{OT1}{pzc}{m}{it}

\numberwithin{equation}{section}
\newcommand{\ii}{{\rm i}}
\newcommand{\ee}{{\rm e}}
\newcommand{\x}{{\rm x}}

\newcommand{\vol}{{\rm vol}}

\newcommand{\dd}{{\rm d}}
\newcommand{\R}{\mathcal{R}}
\DeclareMathOperator{\supp}{supp}
\newcommand{\1}{1\!\!1}
\newcommand{\beq}{\begin{equation}}
\newcommand{\ene}{\end{equation}}

\newtheorem{thm}{Theorem}
\newtheorem{lemma}[thm]{Lemma}
\newtheorem{prop}[thm]{Proposition}

\newtheorem{rem}[thm]{Remark}

\theoremstyle{definition}
\newtheorem{defn}[thm]{Definition}
\newtheorem{ex}[thm]{Example}

\newdateformat{daymonthyear}{\THEDAY \, \monthname[\THEMONTH] \THEYEAR}
\begin{document}

\title{Semiclassical gravity in static spacetimes \\ as a constrained initial value problem}
\author{Benito A. Ju\'arez-Aubry}
\affil{Departamento de Gravitaci\'on y Teor\'ia de Campos, Instituto de Investigaciones Nucleares, Universidad Nacional Aut\'onoma de M\'exico, A. Postal 20-126, CDMX, M\'exico}
\affil{benito.juarez@correo.nucleares.unam.mx}
\date{\daymonthyear\today}

\maketitle

\begin{abstract}
We study the semiclassical Einstein field equations with a Klein-Gordon field in ultrastatic and static spacetimes. In both cases, the equations for the spacetime metric become constraint equations. In the ultrastatic case, the Hadamard singular structure can be characterised explicitly, which allows one to in principle give initial data for the Wightman function that has correct distributional singularities, such that the expectation value of the renormalised stress-energy tensor of the solution can be defined, and that hence the semiclassical Einstein equations make sense. Assuming a ``positive energy" condition for the Klein-Gordon operator, we characterise the states for which, if the constraints hold for initial data, they hold everywhere in spacetime. These turn out to be time-translation invariant states. The static case is analysed by conformal techniques, effectively reducing the problem to an ultrastatic one.
\end{abstract}

%\centerline{\textbf{Private communication. Not for circulation.}}

\section{Introduction}
\label{sec:Intro}

Semiclassical gravity describes the coupled dynamics of quantum matter with the classical spacetime metric, where the latter is sourced by the expectation value of the renormalised stress-energy tensor of the quantum fields. Thus, the theory describes both how spacetime curves by the effects of the energy, pressure and stress of quantum fields in (generally curved) spacetime, and the effects of spacetime curvature on the quantum fields dynamics. 

It is likely that semiclassical gravity is at least relevant as a semiclassical regime of quantum gravity, sufficiently far away from Planck scale. The absence of a fully-workable and well-defined theory of quantum gravity, despite substantial progress in constructing it, adds importance to the understanding of the semiclassical regime. (A number of references to the different approaches to quantum gravity appear in sec. 1 of \cite{DeWitt:2007mi}, a reprint of DeWitt's 1978 Carg\`ese lectures \cite{DeWitt:1978gi}.) In any case, certain situations in theoretical physics in which quantum effects are important and the spacetime curvature cannot be neglected are described in the semiclassical gravity setting, including black hole physics and cosmology. The review \cite{Ford} discusses these and other important situations, as well as other relevant applications and effects in semiclassical gravity, and provides further valuable references in this direction. 

We should also stress that some of the most important open problems in physics are semiclassical gravity problems. Worth of emphasis is the information loss puzzle \cite{Hawking:1976ra, Hawking:2005kf}, which suggests that it could be the case that a low entropy state for the quantum matter would increase its entropy after the evaporation of a black hole has been completed. It would seem that an uncontroversial resolution of this puzzle involves understanding the coupled dynamics of an evaporating black hole, which induces Hawking radiation from the quantum fields in spacetime, while the stress-energy tensor of the quantum field drives the black hole evaporation due to a negative ingoing flux of stress-energy at the black hole horizon. We point the reader to some very recent and interesting perspectives on the information loss puzzle \cite{Arrechea:2020jtv, Kay:2019ukr, Maudlin:2017lye, Modak:2014vya, Unruh:2017uaw, Wallace:2017wzs}.

Compared to general relativity with classical matter, the field equations of semiclassical gravity take a more complicated form even if only a free quantum field is present in the matter sector, see e.g. eq. \eqref{semiEFEKG1} below for a Klein-Gordon field playing the r\^ole of matter. The reason for this is that the right-hand side of the Einstein equation (eq. \eqref{semiEFE1} below) is now defined in terms of a renormalised quantity, since the stress-energy tensor is quadratic in the field. The renormalisation of the stress-energy tensor introduces (i) higher-order terms for the metric tensor that cast the problem as a fourth order problem for the metric, (ii) renormalisation ambiguities, which should be presumably fixed by experiments or quantum gravity, and (iii) requires a precise control of the distributional singularities of the correlation functions of the matter fields in spacetime in the state where one wishes to define the expectation value of the stress-energy tensor. That the singular structure of the state be correct (of Hadamard form for free fields), and hence that the renormalised stress-energy tensor exist, poses already a problem for prescribing ``good" initial data for the matter fields. Indeed, how to guarantee that the Cauchy \emph{initial data} for the matter fields will yield a state, as defined by \emph{spacetime correlations}, for which the expectation value of the stress-energy tensor exists is an important open question in semiclassical gravity. The upcoming work \cite{Programmatic} shows that it is in principle possible to prescribe initial data for a Klein-Gordon field for which the resulting state will satisfy a weak version of the Hadamard condition that allow one to define the renormalised stress-energy tensor.

Despite the obstacles above mentioned, substantial efforts and progress have been made in understanding semiclassical gravity. In the works \cite{Parker:1993dk} and \cite{Flanagan:1996gw} it has been studied how to perturbatively reduce the fourth order problem to a second order problem with no runaway solutions, based on a \emph{perturbative agreement principle} introduced by Simon in \cite{Simon:1990ic}.% The upcoming \cite{Programmatic} is also in this vein, but includes new viewpoints and a discussion on the initial data of the semiclassical gravity.

More recently, there have been works that study semiclassical cosmology quite thoroughly, including the existence and uniqueness of solutions, and make use the modern advancements in the understanding of quantum field theory in curved spacetimes. The treatment of cosmology is more amenable than that of full semiclassical gravity because all of the spacetime geometry is encoded in the scale factor of the metric tensor. We refer the reader to \cite{Dappiaggi:2008mm, Eltzner:2010nx, Gottschalk:2018kqt, Meda:2020smb, Pinamonti:2010is, Pinamonti:2013wya} in this direction. The existence of solutions of semiclassical gravity with conformally covariant fields in globally hyperbolic spacetimes has been studied very recently in \cite{Juarez-Aubry:2021abq}, where the well-posedness in conformally static spacetimes is also studied.

We should add that the study of semiclassical gravity as an initial value problem has been the motivation of our previous work \cite{Juarez-Aubry:2019jon}, where a semiclassical toy model of two scalar fields is solved perturbatively in detail, highlighting some difficulties of the full semiclassical gravity theory, but such that also lessons can be drawn from such simplified system. Semiclassical gravity as an initial value problem is also of relevance for the generally covariant generalisation of a programme in the foundations of physics, which seeks to address the measurement problem by introducing \emph{objective state collapse theories}. In the context of semiclassical gravity we refer to \cite{Modak:2014vya, Canate:2018wtx, DiezTejedor:2011ci, Juarez-Aubry:2017ery, Maudlin:2019bje, Okon:2016qlh, Tilloy:2015zya} for details on the motivation, advances and details.  In the case of our previoius work \cite{Juarez-Aubry:2017ery} an extended introduction can be found in the arXiv v1 preprint. The effects of the state collapse on spacetime in the semiclassical regime are studied in the upcoming \cite{Programmatic}.

The purpose of this paper is to study semiclassical gravity in an amenable situation, namely in general \emph{static spacetimes}, as an initial value problem. The amenability comes from the fact that, in this context, the time-translation symmetry of spacetime renders the dynamical field equations for the metric tensor into constraint equations for the initial data of the quantum field state, which should then be preserved along the state's evolution. Particular cases of semiclassical gravity in static situations have been studied in \cite{Juarez-Aubry:2019jbw} and \cite{Sanders:2020osl}. \cite{Juarez-Aubry:2019jbw} is concerned with the study of quantum field theory in de Sitter spacetime and has connection with the cosmological constant problem, while \cite{Sanders:2020osl} studies a static spacetime where the spatial section is the 3-sphere, motivated by positive-curvature FLRW spacetimes in cosmology. In the context of this paper, the problem studied in \cite{Sanders:2020osl} is the case in which the spacetime is ultrastatic with the spatial section of spacetime corresponding to the 3-sphere.

There are at least two important motivations for studying semiclassical gravity in the static situation: 

First, it is in static (or more generally in stationary) spacetimes where one can define distinguished states for the matter fields. In the case of the ground state, this owes to the existence of a distiguished notion of time -- time elapses as the integral curve of the global timelike, irrotational Killing vector field --, which allows for a distiguished positive- and negative- frequency split for the quantum field representation in terms of creation and annihilation operators in Fock space. The concrete creation and annihilation operators allow one to define a distinguished Hamiltonian for the theory and a vacuum state. (In the Bosonic case this is equivalent to the existence of a distinguished complex structure for the complexified space of classical solutions of the theory.) Positive-temperature equilibrium states are obtained by the mixing of positive and negative frequency modes in such a way that the KMS condition is satisfied. 

Second, as we have mentioned above, characterising the admissible initial data for the quantum state in full semiclassical gravity is a challenging problem. In static spacetimes, as we shall see, this can be done explicitly, in such a way that the resulting state for the solution of the matter equation will have the correct singular structure for the renormalised stress-energy tensor  to make sense everywhere in spacetime (and hence for the semiclassical Einstein field equations to actually make sense too). Therefore, a strong motivation for this work is that understanding the static situation opens the possibility of prescribing good initial data for spacetimes that have a static region (which can be later ignored to obtain data for general spacetimes).

The strategy that we shall follow to analyse semiclassical gravity in static spacetime is to first analyse the ultrastatic case, and then use conformal techniques to handle the static case, by viewing static spacetimes as conformally related to ultrastatic ones with a time-independent conformal factor and showing that the singular structure for states in the static spacetimes is Hadamard if they are conformally related to Hadamard states in the ultrastatic spacetime (under mild assumptions for the conformal factor).

After a brief review of semiclassical gravity in sec. \ref{sec:Review}, which will be useful in part to set our notation, sec. \ref{sec:Ultrastat} and \ref{sec:States} are dedicated to the ultrastatic analysis for semiclassical gravity with a Klein-Gordon field. Sec. \ref{sec:Ultrastat} is devoted to the analysis of the constraints on initial data imposed by the ultrastatic semiclassical Einstein field equations. This involves in particular an analysis of the Hadamard singular structure of physical states, which can be obtained in terms of the spacetime $3$-metric on the initial value Cauchy surface due to the spacetime symmetry. This allows one to control the state singularities completely and guarantee the existence of the expectation value of the renormalised stress-energy tensor, so that the field equations make sense. Sec. \ref{sec:States} then classifies the states for which, if the initial value constraints hold, then they are preserved throughout spacetime. Under a ``positive energy" assumption, these states will be those which are time-translation invariant, which include equilibrium states, such as the vacuum and KMS states among others (provided they satisfy the initial constraints). We show that there are no solutions for coherent states in static semiclassical gravity.

The static case is then handled in sec. \ref{sec:Stat} with the aid of conformal techniques, since static spacetimes are conformally related to ultrastatic ones. The states obtained by such conformal techniques are shown to be Hadamard and the Hadamard singular structure can be obtained by conformal techniques, as well as the constraints and constraint-preserving states. Some examples in static spacetimes and a no-go result are presented in sec. \ref{sec:Examples}. In particular, the no-go result states that the only Ricci-flat, static spacetime that admits semiclassical gravity solutions is Minkowski spacetime. The examples given are for semiclassical gravity in ultrastatic spacetimes with hyperbolic spatial section, in de Sitter spacetime and in a locally flat spacetimes with toroidal spatial section. A summary, final remarks and perspectives of this work are presented in sec. \ref{sec:Conc}.

We follow the convention that a spacetime is a connected, paracompact, differentiable (i.e., smooth) manifold equipped with a Lorentzian metric of signature $(-+++)$, which is in addition time orientable. We also use units in which $c = \hbar = 1$, and use abstract index notation whenever convenient.

\section{A review on semiclassical gravity}
\label{sec:Review}

The equations of semiclassical gravity with a Klein-Gordon matter field with mass $m^2 \geq 0$ and curvature coupling $\xi \in \mathbb{R}$ take the form
\begin{subequations}
\label{semiEFEKG1}
\begin{align}
& G_{ab} + \Lambda^{\rm b} g_{ab} = 8 \pi G_{\rm N}^{\rm b} \omega (T_{ab}) \label{semiEFE1}, \\
& (\Box - m^2 - \xi R) \Phi = 0, \label{KG1}
\end{align}
\end{subequations}
where $\omega$ is a Hadamard state, i.e. the integral kernel of its Wightman two-point function takes the following local form in a convex normal neighbourhood,
\begin{align}
G^+(\x, \x') := \omega(\Phi(\x) \Phi(\x')) = \frac{1}{2 (2 \pi)^2} \left( \frac{\Delta^{1/2}(\x, \x')}{\sigma_\epsilon(\x, \x')} +  v(\x, \x') \ln \left(\frac{\sigma_\epsilon(\x, \x')}{\ell^2} \right) + w_{{\ell} }(\x, \x') \right),
\label{HadamardCondition}
\end{align}
where $\Delta$ is the van Vleck-Morette determinant, $\sigma_\epsilon(\x, \x') := \sigma(\x, \x') + \ii \epsilon (t(\x)-t(\x') ) + \frac{1}{2}\epsilon^2$ is the regularised half-squared geodesic distance with $t$ an arbitrary time function, $v$ and $w$ are smooth and symmetric coefficients and $\ell \in \mathbb{R}$ is an arbitrary length scale. $v$ is completely fixed (at least formally) by expressing it as a covariant Taylor series
\begin{align}
v(\x, \x') = \sum_{n = 0}^\infty v_n(\x, \x') \sigma^n(\x, \x'), \label{vHadamard}
\end{align}
with the $v_n$ coefficients obtained from the so-called Hadamard recursion relations. $v$ turns out to depend only on the spacetime geometry and the parameters $m^2$ and $\xi$. The coefficient $w_\ell$ encodes the details of the quantum state, i.e., for fixed $\ell$, two different Hadamard states will differ only in the coefficient $w_\ell$.

How to obtain the stress-energy tensor on the right-hand side of eq. \eqref{semiEFE1} by a point-splitting regularisation and Hadamard renormalisation method is well known. One must subtract from the two-point function \eqref{HadamardCondition} the Hadamard bi-distribution, $H_\ell$, which in a convex normal neighbourhood takes the form
\begin{align}
H_\ell(\x, \x') = \frac{1}{2 (2 \pi)^2} \left( \frac{\Delta^{1/2}(\x, \x')}{\sigma_\epsilon(\x, \x')} +  v(\x, \x') \ln \left(\frac{\sigma_\epsilon(\x, \x')}{\ell^2} \right) + w^0_{{\ell} }(\x, \x') \right),
\label{Hadamard}
\end{align}
with $w_\ell^0$ arbitrary but fixed (at least formally) by expressing it as a covariant Taylor series analogous to \ref{vHadamard},
\begin{align}
w_\ell^0(\x, \x') = \sum_{n = 0}^\infty w_{\ell \, n}^0(\x, \x') \sigma^n(\x, \x'), \label{wHadamard}
\end{align}
with the coefficients $w_{\ell \, n}^0$ determined by the Hadamard recursion relations up to the arbitrary choice of the ``boundary" term of the series $w_{\ell \, 0}^0$. For the purposes of renormalisation a common choice is to set $w_\ell^0(\x, \x') = 0$, %$w_{\ell \, 0}^0 = 0$, 
and we will make this choice here, %In this case, it follows from the Hadamard recursion relations (see e.g. \cite[eq. (39)]{Decanini:2005eg}) that \begin{align}
%\lim_{\x' \to \x} w_{\ell, 1}^0(\x, \x') = 0, \quad \text{ if } \quad w_{\ell, 0}^0(\x, \x') = 0,
%\label{w1}
%\end{align}
so that the regular part of $H_\ell$ in eq. \eqref{Hadamard} does not contribute to the stress-energy renormalisation.

For Hadamard states, the difference $G^+ - H_\ell$ is free from distributional singularities as $\x' \to \x$. In order to define the expectation value of the stress-energy tensor, one applies to such regular difference the point-splitted operator \cite{Moretti:2001qh}
\begin{subequations}
\begin{align}
\mathscr{T}_{ab} &  := \mathcal{T}_{ab} - \frac{1}{3}g_{ab} \left( g^{a'b'} \nabla_{a'} \nabla_{b'} - m^2 - \xi  R \right), \label{TabPointSplit1} \\
\mathcal{T}_{ab} &  := (1-2\xi ) g_{b}\,^{b'}\nabla_a \nabla_{b'} +\left(2\xi - \frac{1}{2}\right) g_{ab}g^{cd'} \nabla_c \nabla_{d'}  - \frac{1}{2} g_{ab} m^2 + 2\xi \Big[  - g_{a}\,^{a'} g_{b}\,^{b'} \nabla_{a'} \nabla_{b'} + g_{ab} g^{c d}\nabla_c \nabla_d + \frac{1}{2}G_{ab} \Big],
\label{TabPointSplit2}
\end{align}
\end{subequations}
where $g_a{}^{a'}$ is the parallel-transport propagator, and one proceeds to take the limit $\x' \to \x$ to the resulting expression. Since $H_\ell$ is not a solution to the Klein-Gordon equation in its second argument, the second term on the right-hand side of eq. \eqref{TabPointSplit1} contributes to the expectation value of the renormalised stress-energy tensor as 
\begin{align}
& - \frac{1}{3}g_{ab} \lim_{\x' \to \x}    \left( \Box_{\x'} - m^2 - \xi  R  \right)  \left[G^+( \x, \x' )- H_\ell(\x,\x') \right] = \frac{1}{ (2 \pi)^2} g_{ab} [v_1](\x) \nonumber \\
& =  \frac{1}{(2 \pi)^2} g_{ab} \left( \frac{1}{8} m^4 + \frac{1}{4} \left(\xi - \frac{1}{6}\right) m^2 R - \frac{1}{24}\left(\xi - \frac{1}{5} \right) \Box R   + \frac{1}{8} \left(\xi - \frac{1}{6} \right)^2 R^2 - \frac{1}{720} R_{ab} R^{ab} + \frac{1}{720} R_{abcd} R^{abcd} \right),
\label{TraceAnom}
\end{align}
where $[v_1](\x) := \lim_{\x' \to \x} v_1(\x, \x')$ is the diagonal of the coefficient $v_1$ in the Hadamard expansion of $v$, cf. eq. \eqref{vHadamard}. For a conformally coupled scalar the term \eqref{TraceAnom} gives rise to the trace anomaly of the quantum stress-energy tensor.

The above-described prescription contains some ambiguities. Indeed, already the definition of eq. \eqref{Hadamard} contains an arbitrarily length scale $\ell$ and an arbitrary fixed coefficient $w^0_\ell$ in \eqref{wHadamard}. Such ambiguities can be accounted for by adding to the renormalisation prescription an ambiguity term of the form
\begin{subequations}
\label{Theta}
\begin{align}
\Xi_{ab} & :=  \alpha_1 g_{ab} + \alpha_2 G_{ab} + \alpha_3 I_{ab} + \alpha_4 J_{ab}, \\
I_{ab} & :=  - \Box R_{ab} + R_{;ab} - 2 R^{cd} R_{cadb} - \frac{1}{2} g_{ab} \Box R + \frac{1}{2} g_{ab} R_{cd} R^{cd} , \\
J_{ab} & := 2 R_{;ab} - 2 g_{ab} \Box R + \frac{1}{2}g_{ab} R^2 - 2 R R_{ab},
\end{align}
\end{subequations}
with arbitrary $\alpha_i \in \mathbb{R}$ ($i = 1, \ldots 4$), to the final expression.

%{\color{red}
%\begin{align}
%I_{ij} & := {\color{red} - \Delta_h \R_{ij} + \R_{;ij} - 2 \R^{kl} \R_{kilj} - \frac{1}{2} h_{ij} \Delta_h \R + \frac{1}{2} %h_{ij} \R_{kl} \R^{kl} } \\
%g^{ab} J_{ab} & := 2 \Box R - 8 \Box R + 2  R^2 - 2 R^2 = - 6 \Box R,
%\end{align}
%}

%In view of eq. \eqref{w1}, one can use a subtraction renormalisation prescription that preserves the virtues of the Hadamard subtraction, but sets the $w_\ell^0$ term in \eqref{Hadamard} to zero altogether, as this does not change the result for the expectation value of the renomalised stress-energy tensor. Defining
%\begin{equation}
%H_\ell^{\rm sing}(\x, \x') := H_\ell(\x, \x') - \frac{1}{2(2 \pi)^2} w_\ell^0(\x, \x'),
%\label{Hsing}
%\end{equation}
%we have that $G^+ - H_\ell^{\rm sing} = w_\ell$. The price to pay is that $H_\ell^{\rm sing}$ is no longer a solution to the homogeneous Klein-Gordon equation in the first argument.

In summary, the expectation value of the stress-energy tensor appearing on the right-hand side of eq. \eqref{semiEFE1} is
\begin{align}
\omega(T_{ab}(\x)) & = \lim_{\x' \to \x} \mathcal{T}_{ab}(G^+(\x, \x') - H_\ell(\x, \x')) + \frac{1}{8 \pi^2} g_{ab} [v_1](\x) + \Xi_{ab}(\x) \nonumber \\
& = \lim_{\x' \to \x} \mathcal{T}_{ab}w_\ell(\x, \x') + \frac{1}{(2 \pi)^2} g_{ab} [v_1](\x) + \Xi_{ab}(\x),
\label{TabSummary}
\end{align}
where the second line holds for our choice $w_{\ell}^0 = 0$ in the definition of $H_\ell$.

Note that the coefficients $\alpha_1$ and $\alpha_2$ can be seen as renormalising the ``bare" cosmological and Newton's constant, $\Lambda^{\rm b}$ and $G_{\rm N}^{\rm b}$ respectively. We can eliminate them from \eqref{semiEFE1} by defining 
\begin{subequations}
\begin{align}
\Lambda & :=  \frac{\Lambda^{\rm b} - 8 \pi G_{\rm N}^{\rm b} \alpha_1}{1 -  8 \pi G_{\rm N}^{\rm b} \alpha_2} = \Lambda^{\rm b} + O(\ell_{\rm P}^2), \quad 
& G_{\rm N}  & :=  \frac{G_{\rm N}^{\rm b}}{1 -  8 \pi G_{\rm N}^{\rm b} \alpha_2} = \ell_{\rm P}^2 + O(\ell_{\rm P}^4), \label{RenConstants} \\
\alpha & :=  \frac{8 \pi G_{\rm N}^{\rm b} \alpha_3}{1 -  8 \pi G_{\rm N}^{\rm b} \alpha_2} = O(\ell_{\rm P}^2), \quad
& \beta & :=  \frac{8 \pi G_{\rm N}^{\rm b} \alpha_4}{1 -  8 \pi G_{\rm N}^{\rm b} \alpha_2} = O(\ell_{\rm P}^2), \label{NewConstants}
\end{align}
\end{subequations}
where in each case the equalities on the right-hand side are expressed in natural units ($\hbar = 1, c = 1$).\footnote{One sees from eq. \eqref{RenConstants} that the ``renormalised" cosomological and Newton's constant are small corrections from their ``bare" values and from \eqref{NewConstants} that the new constants $\alpha$ and $\beta$ are largely suppressed and only become relevant for curvatures of order $O(1/\ell_{\rm P}^2)$, i.e., presumably in the quantum gravity regime.}

For semiclassical problems such as \eqref{semiEFEKG1}, it is most convenient to trade the matter field equation \eqref{KG1} by equations on the Wightman function, as argued in \cite{Juarez-Aubry:2019jon}. Upon this trade, we demand that the Wightman function $G^+(\x, \x')$ satisfy the Klein-Gordon equation in each of its arguments. Thus, the system of equations \eqref{semiEFEKG1} can be written as
\begin{subequations}
\label{semiEFEKG2}
\begin{align}
& G_{ab} + \Lambda g_{ab} = 8 \pi G_{\rm N} \lim_{\x' \to \x} \mathcal{T}_{ab} \left(G^+(\x, \x') - H_\ell (\x, \x') \right) + \frac{ 2 G_{\rm N}}{\pi} g_{ab} [v_1] + \alpha I_{ab} + \beta J_{ab}, \label{semiEFE2}\\
& (\Box - m^2 - \xi R(\x)) G^+(\x, \x') = 0 = (\Box' - m^2 - \xi R(\x')) G^+(\x, \x'). \label{KG2}
\end{align}
\end{subequations}

For states with non-vanishing one-point function, one should require that the Klein-Gordon equation be satisfied for the one-point function too. For the most part, in our analysis we will assume the one-point function vanishes, but we relax this assumption in sec. \ref{subsec:Coherent}. Eq. \eqref{semiEFE2} can be also written as
\begin{align}
G_{ab} + \Lambda g_{ab} = 8 \pi G_{\rm N} \lim_{\x' \to \x} \mathcal{T}_{ab} w_\ell(\x, \x') + \frac{2 G_{\rm N}}{\pi} g_{ab} [v_1] + \alpha I_{ab} + \beta J_{ab}.
\label{semiEFE2-alt}
\end{align}

As a Cauchy problem, the system \eqref{semiEFE2} should provide solutions $g_{ab}$ and $G^+$ in a globally hyperbolic spacetime, given initial data, if the problem is well posed. Since the system \eqref{semiEFE2} is second order for the matter and fourth order for the metric, the initial data provided in a would-be Cauchy surface, $S$, with normal $n^a$, should consist of data $h_{ab}, K_{ab}, \nabla_n K_{ab}$ and $\nabla_n \nabla_n K_{ab}$ for $g_{ab}$ and data %$\omega( \varphi(\x) \varphi(\x') )$, $\omega( \pi(\x) \varphi(\x') ), \omega( \varphi(\x) \pi(\x') )$ and $\omega( \pi(\x) \pi(\x') )$
\begin{subequations}
\label{WightmanIVP}
\begin{align}
G^+(\x,\x')|_S = \omega(\varphi(\underline{\x}) \varphi(\underline{\x}')) = G^+_{\varphi\varphi}(\underline{\x}, \underline{\x}'), & \quad & \nabla_n G^+(\x,\x')|_S = \omega(\pi(\underline{\x}) \varphi(\underline{\x}')) = G^+_{\pi \varphi}(\underline{\x}, \underline{\x}'), \\
\nabla_{n'}G^+(\x,\x')|_S = \omega(\varphi(\underline{\x}) \pi(\underline{\x}')) = G^+_{\varphi \pi}(\underline{\x}, \underline{\x}'), & \quad & \nabla_n \nabla_{n'} G^+(\x,\x')|_S = \omega(\pi(\underline{\x}) \pi(\underline{\x}')) = G^+_{\pi \pi}(\underline{\x}, \underline{\x}'),
\end{align}
\end{subequations}
for $G^+(\x, \x')$, where $\varphi$ and $\pi$ are the 3-field and its momentum on $S$ (and $\underline{\x}, \underline{\x}' \in S$). Naturally, the canonical commutation relations must hold for the state on $S$, so that the following constraint holds distributionally
\begin{align}
\omega( \varphi(\underline{\x}) \pi(\underline{\x}') ) -\omega( \pi(\underline{\x}) \varphi(\underline{\x}') ) = \ii \delta_h(\underline{\x}, \underline{\x}').
\end{align}

A number of difficulties in showing that \eqref{semiEFEKG2} is well posed are well-known. The aim of this paper is to show that unique solutions exist to the system \eqref{semiEFEKG2} in the (ultra)static case for suitable initial data.

%{\color{red} CONTINUE HERE!}

\section{Semiclassical gravity in ulstrastatic spacetimes}
\label{sec:Ultrastat}

In this section, we study the system \eqref{semiEFEKG2} as an initial value problem in ultrastatic situations. Assume that the spacetime metric takes the local form
\begin{align}
g_{ab} = - \dd t_a \otimes \dd t_b + h_{ij} \dd x^i_a \otimes \dd x^j_b, 
\label{UltraStat-metric}
\end{align}
in local coordinates $p^\mu = (t, x^i)$ and that $\eta^a = \partial_t^a$ is a Killing vector field. In this case, the Riemann tensor of the problem satisfies locally  that
\begin{align}
R_{\mu \nu \rho}{}^{\sigma} =\left\{
                \begin{array}{l}
 0, \text{ if any coordinate index is equal to } t, \\
 \mathcal{R}_{i j k}{}^{l}, \text{ with }  i,j,k,l = 1, \ldots 3 \text{ otherwise}. 
\end{array}\right.
\label{UltraStat-Riemann}
\end{align}
where $\mathcal{R}_{abcd}$ denotes the 3-Riemann tensor of the 3-metric $h_{ab}$. We shall use similar caligraphic notation for the Ricci tensor and scalar, $\mathcal{R}_{ac} = \mathcal{R}_{abc}{}^b$ and $\mathcal{R} = h^{ab} \mathcal{R}_{ab}$ respectively. That the Riemann tensor takes the form of eq. \eqref{UltraStat-Riemann} can be seen directly from the local definition of the Riemann tensor in terms of connection coefficients. 

Set the initial value surface $S$ at $t = 0$ defined by Killing time. We assume that the spacetime is globally hyperbolic, thus $S$ is geodesically complete \cite[p. 130]{Fulling:1989nb} (and hence complete in the metric-space sense). $\eta^a$ is surface-orthogonal, and the extrinsic curvature of $S$, $K_{ab}$, and any number of its normal derivatives vanish identically. Moreover, the half-squared geodesic distance takes the form $\sigma = - \frac{1}{2}(t-t')^2 + \sigma_h$, where $\sigma_h$ is the half-squared geodesic distance in the spacetime $(S, h)$.

In the coordinates $p^\mu = (t, x^i)$, the semiclassical gravity equations \eqref{semiEFEKG2} take the simplified form
\begin{subequations}
\label{semiEFEKG3}
\begin{align}
&  \frac{1}{2} \R - \Lambda = 8 \pi G_{\rm N} \lim_{p' \to p} \mathcal{T}_{tt} \left(G^+(p, p') - H_\ell (p, p') \right) - \frac{2 G_{\rm N}}{\pi}  [v_1] + \frac{\alpha}{2} \left(\Delta_h \R - \R^{ij}\R_{ij}\right) + \beta \left(2 \Delta_h \R - \frac{1}{2} \R^2 \right), \label{semiEFE3-1}\\
& \lim_{p' \to p} \mathcal{T}_{t i} \left(G^+(p, p') - H_\ell (p, p') \right) = 0 \label{semiEFE3-2}\\
& \R_{ij} - \frac{1}{2} \R h_{ij} + \Lambda h_{ij} = 8 \pi G_{\rm N} \lim_{p' \to p} \mathcal{T}_{ij} \left(G^+(p, p') - H_\ell (p, p') \right) + \frac{2 G_{\rm N}}{\pi} h_{ij} [v_1] + \alpha I_{ij} + \beta J_{ij}, \label{semiEFE3-3}\\
& (\partial_t^2 + A_h) G^+ = (\partial_{t'}^2 + A_{h'}) G^+ = 0, \quad \text{with} \quad A_h := -\Delta_h + m^2 + \xi \R,  \label{KG3}
\end{align}
\end{subequations}
where $\Delta_h $ is the Laplacian operator associated with the 3-metric $h_{ab}$ and
\begin{subequations}
\begin{align}
 [v_1] & =   \frac{1}{8} m^4 + \frac{1}{4} \left(\xi - \frac{1}{6}\right) m^2 \R - \frac{1}{24}\left(\xi - \frac{1}{5} \right) \Delta_h \R   + \frac{1}{8} \left(\xi - \frac{1}{6} \right)^2 \R^2 - \frac{1}{720} \R_{ij} \R^{ij} + \frac{1}{720} \R_{ijkl} \R^{ijkl} , \label{UltraStat-v1} \\
I_{ij} & =  - \Delta_h \R_{ij} + \R_{;ij} - 2 \R^{kl} \R_{kilj} - \frac{1}{2} h_{ij} \Delta_h \R + \frac{1}{2} h_{ij} \R_{kl} \R^{kl} , \label{UltraStat-Iij}\\
J_{ij} & = 2 \R_{;ij} - 2 h_{ij} \Delta_h \R + \frac{1}{2}h_{ij} \R^2 -  2 \R \R_{ij}. \label{UltraStat-Jij}
\end{align}
\end{subequations}

If $H_\ell(p,p')$ is known, then eq. \eqref{semiEFE3-1}, \eqref{semiEFE3-2} and \eqref{semiEFE3-3} yield initial-data constraints when evaluated at $t = 0$ for the Wightman two-point function data that can be handled explicitly. Such constraints should be preserved at all times. We shall now see in sec. \eqref{subsec:Had-Sing} that indeed $H_\ell(p,p')$ can be obtained in closed form (up to any desired order in $\sigma$) in the ultrastatic case. We shall furthermore analyse these constraints in more detail below in sec. \eqref{subsec:Constraints}.

\subsection{The Hadamard singular structure in ultrastatic spacetimes}
\label{subsec:Had-Sing}

We now show that in ultrastatic spacetimes the Hadamard singular structure is known to any desired order in $\sigma$ for arbitrary 3-metric $h_{ab}$ (see eq. \eqref{UltraStat-metric}). The requirement that the Klein-Gordon state be Hadamard, cf. eq. \eqref{HadamardCondition}, imposes that in a convex normal neighbourhood
\begin{subequations}
\label{HadamardFormUltrastatic}
\begin{align}
G^+ & = \frac{1}{2 (2 \pi)^2} \left( \frac{\Delta^{1/2}}{\sigma_\epsilon} +  v \ln \left(\frac{\sigma_\epsilon}{\ell^2} \right) + w_{{\ell} }\right), \text{ with }\\
\sigma(p,p') & = -\frac{1}{2}(t-t')^2 + \sigma_h(x,x'), \label{UltraStat-sigma} \\
\Delta(p,p') & = -[-\det g(p)]^{-1/2} \det[\sigma_{;a b'}(p,p') ] [-\det g(p')]^{-1/2} \nonumber \\
& = -[\det h(x)]^{-1/2} \det[\sigma_{h;ij'}(p,p') ] [\det h(x')]^{-1/2} = - \Delta_{(3)}(x,x'), \label{UltraStat-VanVleck}\\
v(p,p') & = v_0(x,x') + v_1(x,x') \sigma(p,p') + O(\sigma^{3}), \label{UltraStat-vcoeff}
\end{align}
\end{subequations}
where $\Delta_{(3)}$ in eq. \eqref{UltraStat-VanVleck} is the Van Vleck-Morette determinant of the 3-metric $h_{ab}$. The minus sign on the right-hand side of eq. \eqref{UltraStat-VanVleck} accounts for the fact that $h_{ab}$ has Euclidean signature. The coefficients $v_n$ obeys the Hadamard recursion relations, see e.g. \cite[eq. (38)]{Decanini:2005eg}. On the right-hand side of eq. \eqref{UltraStat-vcoeff}, we have written that the $v_n$ coefficients are time-independent. We now explain why this is the case.

It is a direct consequence of the geometric nature of the $v_n$ that they can only depend on time as the difference $t-t'$ (by stationarity). % occuring in the	half-squared geodesic distance \eqref{UltraStat-sigma} or its first covariant derivative, $\sigma^{;\mu} = (-(t-t'), \sigma_h(x,x')^{;i})$. 
Further, each of the coefficients $v_n$ for $v$ (cf. eq. \eqref{vHadamard} and \eqref{UltraStat-vcoeff}) can be expanded in a covariant Taylor series up to any desired order, say $N$, of the form \cite{Decanini:2005eg}
\begin{align}
v_n(p,p') = v_{n0}(p) + \sum_{m = 1}^N v_{n m \, \mu_1 \cdots \mu_m}(p) \sigma^{;\mu_1}(p,p') \cdots \sigma^{;\mu_m}(p,p') + O\left(\sigma^{(N+1)/2}\right).
\label{vnExpansions}
\end{align}

In the ultrastatic case, the coefficient $v_{n0}$ is $t$-independent, since in the chosen coordinates all geometric terms are so. The coefficients $v_{n m \, \mu_1 \cdots \mu_m}$ must be $t$-independent for the same reason. Further, since any tensorial index of the coefficients $v_{n m \, \mu_1 \cdots \mu_m}$ corresponds to the index of a covariant derivative or a curvature tensor index, they are purely spatial. This is so because (a) of time-independence, (b) any connection coefficient containing any $t$-index vanishes and (c) as we have seen in eq. \eqref{UltraStat-Riemann}, the vanishing is also true for the Riemann tensor components containing $t$-indices (as well as its contractions). Thus, eq. \eqref{vnExpansions} can be written as a covariant Taylor series in the spatial geodesic distance, $\sigma_h$, as follows
\begin{align}
v_n(p,p') & = v_{n0}(x) + \sum_{m = 1}^N v_{n m \, i_1 \cdots i_m}(x) \sigma^{;i_1} \cdots \sigma^{;i_m} + O\left(\sigma^{(N+1)/2}\right) \nonumber \\
& = v_{n0}(x) + \sum_{m = 1}^N v_{n m \, i_1 \cdots i_m}(x) \sigma_{h}^{;i_1} \cdots \sigma_h{}^{;i_m} + O\left(\sigma_h^{(N+1)/2}\right) = v_n(x,x'),
\label{vnExpansions2}
\end{align}
where in the second line we used that $\sigma^{;\mu} = (-(t-t'), \sigma_h(x,x')^{;i})$, cf. eq. \eqref{UltraStat-sigma}, so that $\sigma^{;i} = \sigma_h(x,x')^{;i}$.

To illustrate the form that these coefficients take, it suffices for our purposes to quote eq. (108) and (109) of \cite{Decanini:2005eg} adapted to the ultrastatic case
\begin{subequations}
\begin{align}
v_0 &= \frac{1}{2} m^2 + \frac{1}{2}\left( \xi - \frac{1}{6} \right) \R - \frac{1}{4}\left(\xi - \frac{1}{6} \right) \R_i \sigma_h{}^{;i} + \frac{1}{12} \left[ \frac{1}{2} m^2 \R_{ij} + \left(\xi - \frac{3}{29} \right) \R_{;ij} - \frac{1}{20} \Delta_h \R_{ij} \right. \nonumber \\
& \left. + \frac{1}{2} \left( \xi - \frac{1}{6} \right) \R \R_{ij} + \frac{1}{15} \R^c{}_i \R_{cj} - \frac{1}{30} \R^{kl} \R_{kilj} - \frac{1}{30} \R^{klm}{}_i \R_{klmj} \right] \sigma_h{}^{;i} \sigma_h{}^{;j} +O\left( \sigma_h^{3/2} \right)\nonumber \\
v_1 &= [v_1] + O\left(\sigma_h^{1/2}\right),
\end{align}
\end{subequations}
with $[v_1]$ given by eq. \eqref{UltraStat-v1}. 

Hence, as claimed, in ultrastatic spacetimes the Hadamard singular structure can be characterised explicitly in terms of 3-geometric data up to any desired order in the geodesic distance.

\subsection{Constraints on the state initial data}
\label{subsec:Constraints}

Since in the ultrastatic case the metric evolves trivially, eq. \eqref{semiEFE3-1}, \eqref{semiEFE3-2} and \eqref{semiEFE3-3} are constraints for the state of the matter field. The conditions that
\begin{subequations}
\label{Tindep}
\begin{align}
& \frac{\partial}{\partial t} \lim_{p' \to p} \mathcal{T}_{tt} \left(G^+(p, p') - H_\ell (p, p') \right) = 0, \label{Tindep-1}\\
& \frac{\partial}{\partial t} \lim_{p' \to p} \mathcal{T}_{t i} \left(G^+(p, p') - H_\ell (p, p') \right) = 0, \label{Tindep-2}\\
& \frac{\partial}{\partial t} \lim_{p' \to p} \mathcal{T}_{ij} \left(G^+(p, p') - H_\ell (p, p') \right) =0, \label{Tindep-3}
\end{align}
\end{subequations}
guarantee that the constraints hold at all times. A condition on the states that guarantees that eq. \eqref{Tindep-1}, \eqref{Tindep-2} and \eqref{Tindep-3} hold is that they be stationary, i.e., time-translation independent. We shall see below, in sec. \ref{sec:States} that for $A_h = -\Delta_h + m^2 + \xi R > 0$ as an operator in $L^2(S, \dd \vol_h)$, stationarity is not only sufficient, but also necessary for the fulfillment of conditions \eqref{Tindep}.

We now proceed to study constraints \eqref{semiEFE3-1}, \eqref{semiEFE3-2} and \eqref{semiEFE3-3} in more detail. Let us first write them down explicitly by appropriately choosing local coordinates for the parallel-transport propagator. We have that locally
\begin{align}
g^\mu{}_{\mu'}(p,p') = e^{\mu}_I(p) e^I_{\mu'}(p'),
\end{align}
where $e^{\mu}_I$ is the tetrad field with Lorentz indices denoted by capital latin letters and $e_{\mu}^I$ its dual, such that $e^{\mu}_Ie_{\nu}^I = \delta^\mu{}_\nu$ and $e_{\mu}^I e^{\mu}_J = \delta^I{}_J$. Locally
\begin{align}
g_{\mu \nu } = \eta_{IJ} e_\mu^I e_\nu^J, & \quad \eta_{IJ} = g_{\mu \nu} e_I^\mu e_J^\nu, 
\end{align}
which allows us to choose in the ultrastatic case $e_t^T = 1$, $e_t^{X^i} = 0$ and $e_i^T = 0$, as well as $e_T^t = 1$, $e_T^i = 0$ and $e_{X^i}^t = 0$, with the rest of the tetrad components forming the triad soldering the 3-metric components of $g_{\mu \nu}$, $h_{ij}$, to the Euclidean 3-metric. I.e., we choose
\begin{subequations}
\begin{align}
g^{t}{}_{t'}(p,p') & = 1, &   \quad  g^{t}{}_{i'}(p,p') &= 0 , \\
g^{i}{}_{t'}(p,p') &=  0, &  \quad  g^{i}{}_{j'}(p,p') &= h^{i}{}_{j'}(p,p'),
\end{align}
\end{subequations}
where $h^{i}{}_{j'}(p,p')$ is the parallel propagator constructed out of the triad that solders $h_{ij}$ to the Euclidean metric and its dual triad.

The point-splitting operator appearing in eq. \eqref{semiEFE3-1}, \eqref{semiEFE3-2} and \eqref{semiEFE3-3} takes the form
\begin{subequations}
\begin{align}
\mathcal{T}_{tt} &  = (1-2\xi ) \partial_t \partial_{t'} -\left(2\xi - \frac{1}{2}\right) \left( - \partial_t \partial_{t'} + h^{ij'} \nabla_i \nabla_{j'} \right)  + \frac{1}{2} m^2 + 2\xi \Big[  - \partial_t \partial_{t'} -(-\partial_t^2 + \Delta_h) + \frac{1}{4} \R \Big], \\
\mathcal{T}_{ti} &  = (1-2\xi ) h_{i}\,^{j'}\partial_t \nabla_{j'} - 2\xi  h_{i}\,^{j'} \partial_{t'} \nabla_{j'} , \\
\mathcal{T}_{ij} &  = (1-2\xi ) h_{j}\,^{j'}\nabla_i \nabla_{j'} +\left(2\xi - \frac{1}{2}\right) h_{ij}h^{kl'} \nabla_k \nabla_{l'}  - \frac{1}{2} h_{ij} m^2 + 2\xi \Big[  - h_{i}\,^{i'} h_{j}\,^{j'} \nabla_{i'} \nabla_{j'} + h_{ij} h^{kl}\nabla_k \nabla_l + \frac{1}{2}G_{ij} \Big].
\end{align}
\end{subequations}

Evaluating eq. \eqref{semiEFE3-1}, \eqref{semiEFE3-2} and \eqref{semiEFE3-3} on the initial value surface yields the following constraints for the initial data of the Wightman function equation of motion \eqref{KG3},
\begin{subequations}
\label{semiEFEKG-constraints}
\begin{align}
\mathcal{C}_{tt} & :=  \lim_{x' \to x} \left[ (1-2\xi ) \omega(\pi(x) \pi(x')) -\left(2\xi - \frac{1}{2}\right) \left( - \omega(\pi(x) \pi(x')) + h^{ij'} \nabla_i \nabla_{j'} \omega(\varphi(x) \varphi(x')) \right)  + \frac{1}{2} m^2 \omega(\varphi(x) \varphi(x'))\right. \nonumber \\
 & \left. - 2 \xi \left( \Delta_h + \frac{1}{4} \R \right) \omega(\varphi(x) \varphi(x')) - \mathcal{T}_{tt} H_\ell (x, x') \right] - \frac{1}{8 \pi G_{\rm N}} \left[  \frac{1}{2} \R - \Lambda + \frac{2 G_{\rm N}}{\pi}  [v_1] - \frac{\alpha}{2} \left(\Delta_h \R - \R^{ij}\R_{ij}\right) \right.  \nonumber \\
 & \left. - \beta \left(2 \Delta_h \R - \frac{1}{2} \R^2 \right) \right] = 0, \label{semiEFE-constraints1} \\
\mathcal{C}_{ti} & := \lim_{x' \to x}  \left[ (1-2\xi ) h_{i}\,^{j'}\nabla_{j'} \omega(\pi(x) \varphi(x')) - 2\xi  h_{i}\,^{j'} \nabla_{j'} \omega(\varphi(x) \pi(x'))  - \mathcal{T}_{ti} H_\ell (x, x') \right] = 0, \label{semiEFE-constraints2} \\
\mathcal{C}_{ij} & :=  \lim_{x' \to x} \left[ (1-2\xi ) h_{j}\,^{j'}\nabla_i \nabla_{j'} \omega(\varphi(x) \varphi(x')) +\left(2\xi - \frac{1}{2}\right) h_{ij}h^{kl'} \nabla_k \nabla_{l'} \omega(\varphi(x) \varphi(x'))  - \frac{1}{2} h_{ij} m^2 \omega(\varphi(x) \varphi(x')) \right. \nonumber \\  
 & \left. + 2\xi \Big[  - h_{i}\,^{i'} h_{j}\,^{j'} \nabla_{i'} \nabla_{j'} \omega(\varphi(x) \varphi(x')) + h_{ij} h^{kl}\nabla_k \nabla_l \omega(\varphi(x) \varphi(x')) + \frac{1}{2}G_{ij} \omega(\varphi(x) \varphi(x')) \Big]- \mathcal{T}_{ij} H_\ell (x, x') \right] \nonumber \\
 & - \frac{1}{8 \pi G_{\rm N}} \left[ \R_{ij} - \frac{1}{2} \R h_{ij} + \Lambda h_{ij} - \frac{2 G_{\rm N}}{\pi} g_{ij} [v_1] - \alpha I_{ij} - \beta J_{ij} \right] = 0, \label{semiEFE-constraints3}
\end{align}
\end{subequations}
where we have used the Klein-Gordon equation \eqref{KG3} on the initial value surface to eliminate terms of the form $\partial_t \omega(\pi(x) \varphi(x'))$ and $\partial_t' \omega(\varphi(x) \pi(x'))$ in favour of initial data.

Our first observation is that the constraint $\mathcal{C}_{ti} = 0$ \eqref{semiEFE-constraints2} is trivially satisfied, and the stress-energy tensor takes a block-diagonal form for any (not necessarily stationary) Hadamard state in an ultrastatic spacetime, i.e, we can drop eq. \eqref{semiEFE-constraints2} altogether. We now show this.

\begin{lemma}
\label{LemmaCti}
{\rm (a)} The constraint $\mathcal{C}_{ti} = 0$ \eqref{semiEFE-constraints2} is trivially satisfied for any Hadamard state in an ultrastatic spacetime. {\rm (b)} Additionally, if $(h,S)$ is a smooth Riemannian spacetime, then one has the following $\delta$-function representation in $\mathscr{D}'(S)$,
\begin{align}
\delta_h(x, x') = \lim_{\epsilon \to 0^+} - \frac{ \epsilon}{2(2\pi)^2} \left( -\frac{|\Delta_h(x,x')|^{1/2}}{(\sigma_h(x,x') + \frac{1}{2} \epsilon^2)^2} + \frac{v(x,x')}{\sigma_h(x,x') + \frac{1}{2}\epsilon^2}   \right) \label{deltaFormula}
\end{align}
for $x$ and $x'$ in a convex normal neighbourhood.
\end{lemma}
\begin{proof}
To see this, we first use the canonical commutation relations together with Synge's rule $[A]_{;a} = [A_{;a}] + [A_{;a'}]$, where $A$ is a regular bi-tensor and the brackets indicate the standard coincidence-limit notation $[A] := \lim_{\x'\to \x} A$. We have that
\begin{align}
\mathcal{C}_{ti} % & = \lim_{x' \to x}  \left[ (1-2\xi ) h_{i}\,^{j'}\nabla_{j'} [\omega(\pi(x) \varphi(x'))- \partial_{t}H_\ell(x,x')] - 2\xi  h_{i}\,^{j'} \nabla_{j'} [\omega(\varphi(x) \pi(x'))  - \partial_{t'} H_\ell (x, x')] \right] \nonumber \\
& = \lim_{x' \to x}  \frac{1}{2}\left[ (1-2\xi ) h_{i}\,^{j'}\nabla_{j'} [\omega(\pi(x) \varphi(x'))- \partial_{t}H_\ell(x,x')] - 2\xi  h_{i}\,^{j'} \nabla_{j'} [\omega(\varphi(x) \pi(x'))  - \partial_{t'} H_\ell (x, x')] \right] \nonumber \\
& + \partial_t \lim_{x' \to x}  \frac{1}{2}\left[ (1-2\xi ) h_{i}\,^{j'}\nabla_{j'} [\omega(\varphi(x) \varphi(x'))- H_\ell(x,x')] - 2\xi  h_{i}\,^{j'} \nabla_{j'} [\omega(\varphi(x) \varphi(x'))  -  H_\ell (x, x')] \right] \nonumber \\
& - \lim_{x' \to x}  \frac{1}{2}\left[ (1-2\xi ) h_{i}\,^{j'}\nabla_{j'} [\omega(\varphi(x) \pi(x'))- \partial_{t'}H_\ell(x,x')] - 2\xi  h_{i}\,^{j'} \nabla_{j'} [\omega(\pi(x) \varphi(x'))  -  \partial_t H_\ell (x, x')] \right] \nonumber \\
& = \lim_{x' \to x}  \left[ \left(\frac{1}{2}-\xi \right) h_{i}\,^{j'}\nabla_{j'} [- \ii \delta_h (x,x')- (\partial_{t} - \partial_{t'}) H_\ell(x,x')] - \xi  h_{i}\,^{j'} \nabla_{j'} [\ii \delta_h(x,x')  - (\partial_{t'}-\partial_t) H_\ell (x, x')] \right] \nonumber \\
& = \lim_{x' \to x}  \left[ \frac{1}{2} h_{i}\,^{j'}\nabla_{j'} [- \ii \delta_h (x,x')- (\partial_{t} - \partial_{t'}) H_\ell(x,x')]  \right].
\label{CijVanish1}
\end{align}

The derivatives on the last two equalities of eq. \eqref{CijVanish1} are distributional. We now explicitly compute this term. Let us split $H_\ell = H_\ell^{\rm s} + H_\ell^{\rm a}$, where $H_\ell^{\rm s}(p,q) := \frac{1}{2} (H_\ell(p,q) + H_\ell(q,p))$ and $H_\ell^{\rm a}(p,q) := \frac{1}{2} (H_\ell(p,q) - H_\ell(q,p))$ are the symmetric and antisymmetric parts of $H_\ell$ respectively. %Next, note that 
%\begin{align}
%(\partial_{t} - \partial_{t'}) H_\ell(p,p') & = \frac{1}{2 (2 \pi)^2} (\partial_{t} - \partial_{t'}) \left( \frac{\Delta^{1/2}}{\sigma_\epsilon} +  v \ln \left(\frac{\sigma_\epsilon}{\ell^2} \right) \right) \nonumber \\ 
%& = \frac{1}{(2 \pi)^2} \left( -\frac{\Delta^{1/2}}{\sigma_\epsilon^2}(-(t-t') + 2 \ii \epsilon ) -(t-t') \left(\frac{v -v_0}{\sigma} \right) \ln \left(\frac{\sigma_\epsilon}{\ell^2} \right) + \frac{v}{\sigma_\epsilon} (-(t-t') + 2 \ii \epsilon )   \right),
%\end{align}
Choosing the arbitrary regularisation time function in $\sigma_\epsilon$ as $\sigma_\epsilon = \sigma + \frac{\ii}{2} \epsilon(t-t') + \frac{1}{2}\epsilon^2$ we can see that
\begin{subequations}
\begin{align}
(\partial_{t} - \partial_{t'}) & H^{\rm s}_\ell(p,p')  = \frac{1}{2(2 \pi)^2} \left( -\frac{\Delta^{1/2}}{\sigma_\epsilon^2}\left(-(t-t') +  \frac{1}{2}\ii \epsilon \right) -(t-t') \left(\frac{v -v_0}{\sigma} \right) \ln \left(\frac{\sigma_\epsilon}{\ell^2} \right) + \frac{v}{\sigma_\epsilon} \left(-(t-t') +  \frac{1}{2}\ii \epsilon \right)   \right) \nonumber \\
& + \frac{1}{2(2 \pi)^2} \left( -\frac{\Delta^{1/2}}{\sigma_{-\epsilon}^2}\left(-(t-t') - \frac{1}{2}\ii \epsilon \right) -(t-t') \left(\frac{v -v_0}{\sigma} \right) \ln \left(\frac{\sigma_{-\epsilon}}{\ell^2} \right) + \frac{v}{\sigma_{-\epsilon}} \left(-(t-t') -  \frac{1}{2}\ii \epsilon \right)   \right), \\
(\partial_{t} - \partial_{t'}) & H^{\rm a}_\ell(p,p')  = \frac{1}{2(2 \pi)^2} \left( -\frac{\Delta^{1/2}}{\sigma_\epsilon^2}\left(-(t-t') +  \frac{1}{2}\ii \epsilon \right) -(t-t') \left(\frac{v -v_0}{\sigma} \right) \ln \left(\frac{\sigma_\epsilon}{\ell^2} \right) + \frac{v}{\sigma_\epsilon} \left(-(t-t') +  \frac{1}{2}\ii \epsilon \right)   \right) \nonumber \\
& - \frac{1}{2(2 \pi)^2} \left( -\frac{\Delta^{1/2}}{\sigma_{-\epsilon}^2}\left(-(t-t') -  \frac{1}{2}\ii \epsilon \right) -(t-t') \left(\frac{v -v_0}{\sigma} \right) \ln \left(\frac{\sigma_{-\epsilon}}{\ell^2} \right) + \frac{v}{\sigma_{-\epsilon}} \left(-(t-t') -  \frac{1}{2}\ii \epsilon \right)  \right),
\end{align}
\end{subequations}
and thus at $t = t' = 0$ we see that distributionally
\begin{subequations}
\begin{align}
(\partial_{t} - \partial_{t'}) H^{\rm s}_\ell(x,x') & = 0, \\
(\partial_{t} - \partial_{t'}) H^{\rm a}_\ell(x,x') & = \frac{\ii \epsilon}{2(2\pi)^2} \left( -\frac{|\Delta_h|^{1/2}}{(\sigma_h + \frac{1}{2} \epsilon^2)^2} + \frac{v}{\sigma_h + \frac{1}{2}\epsilon^2}   \right). \label{HaInitial}
\end{align}
\end{subequations}

Thus, only the antisymmetric part of the Hadamard singular structure contributes on the initial value surface. $H^{\rm a}$ can be identified immediately. If $\omega$ is a Hadamard state with Wightman two-point function $G^+$, then in a convex normal neighbourhood
\begin{align}
G^+(\x, \x') = \frac{1}{2}G^{(1)}(\x, \x') + \frac{\ii}{2} E(\x, \x') = H^{\rm s}_\ell(\x, \x') + H^{\rm a}_\ell(\x, \x') + \omega_\ell(\x, \x').
\end{align}

Since $G^{(1)}$, $H^{\rm s}_\ell$ and $\omega_\ell$ are symmetric and $H^{\rm a}_\ell$ and $E$ are antisymmetric, we have that $H^{\rm a}_\ell = \frac{\ii}{2} E$. Using the initial value condition for the causal propagator
\begin{align}
-\partial_t E(p, p')|_{t = t' = 0} = \partial_{t'} E(p, p')|_{t = t' = 0} = \delta_h(x, x'),
\end{align}
it follows immediately that
\begin{align}
(\partial_{t} - \partial_{t'}) H^{\rm a}_\ell(x,x') = - \ii \delta_h(x,x').
\label{IdentityDelta}
\end{align}

Inserting eq. \eqref{IdentityDelta} into the right-hand side of eq. \eqref{CijVanish1} we prove part (a). Combining eq. \eqref{HaInitial} with eq. \eqref{IdentityDelta}, we obtain the $\delta$-function formula \eqref{deltaFormula} of part (b).
\end{proof}

An example and further remarks on the $\delta$-function formula appear in App. \ref{Appendix}

%{\color{red} 5-FEB-2021: EDITS UP TO HERE!}

We now proceed to interpret the meaning of constraints \eqref{semiEFE-constraints1} and \eqref{semiEFE-constraints3}. Let us note first that, by manipulations like the ones in Sec. IIC of \cite{Decanini:2005eg} (see especially eq. (71)), the constraints \eqref{semiEFE-constraints1} and \eqref{semiEFE-constraints3} can be written as
\begin{align}
\mathcal{C}_{tt} & = \frac{1}{2(2 \pi)^2} \left[-[w_{\ell,tt}] - \frac{1}{2} \left(2 \xi - \frac{1}{2} \right) \Delta_h [w_\ell] + [v_1] \right] \!- \frac{1}{8 \pi G_{\rm N}} \!\left[  \frac{1}{2} \R - \Lambda - \frac{\alpha}{2} \left(\Delta_h \R  - \R^{ij}\R_{ij}\right) - \beta \left(2 \Delta_h \R - \frac{1}{2} \R^2 \right) \right]\!\! = 0,
\label{Constraint1-DF} \\
\mathcal{C}_{ij} & = \frac{1}{2(2\pi)^2} \left[ -[w_{\ell;ij}] + \frac{1}{2}(1-2\xi) [w_\ell]_{;ij} + \frac{1}{2} \left(2 \xi -\frac{1}{2} \right) h_{ij} \Delta_h [w_\ell] + \xi \R_{ij} [w_\ell] - h_{ij} [v_1] \right] \nonumber \\
& - \frac{1}{8 \pi G_{\rm N}} \left[ \R_{ij} - \frac{1}{2} \R h_{ij} + \Lambda h_{ij} - \alpha I_{ij} - \beta J_{ij} \right] = 0,
\label{Constraint3-DF}
\end{align}
where $[w_\ell] = \lim_{p'\to p} (G^+(p,p')-H_\ell^{\rm sing}(p,p') )$ is the diagonal of the regular part of the state, cf. eq. \eqref{HadamardCondition}, and $[w_{\ell,tt}]$ and $[w_{\ell;ij}]$ are defined analogously. 

Taking the trace of eq. \eqref{Constraint3-DF} with respect to the 3-metric $h_{ij}$, we have
\begin{align}
h^{ij} \mathcal{C}_{ij} & = \frac{1}{2(2\pi)^2} \left[ -[\Delta_h w_\ell] + \left(2 \xi -\frac{1}{4} \right) \Delta_h [w_\ell] + \xi \R [w] - 3 [v_1] \right] \nonumber \\
& - \frac{1}{8 \pi G_{\rm N}} \left[ - \frac{1}{2} \R  + 3 \Lambda - \alpha \left( - \frac{3}{2} \Delta_h \R - \frac{1}{2} \R^{ij}\R_{ij} \right) - \beta \left( - 4 \Delta_h \R - \frac{1}{2} \R^2 \right) \right] = 0. 
\label{TraceCij}
\end{align}

Using eq. (52) in \cite{Decanini:2005eg} we have the relation
\begin{align}
-[w_{\ell,tt}] = [A_h w_\ell] - 6 [v_1] = -[\Delta_h w_\ell] + (m^2 + \xi \R)[w_\ell] - 6 [v_1],
\end{align}
which allows us to eliminate the term $[w_{\ell,tt}]$ in eq. \eqref{Constraint1-DF}, and combine it with the trace eq. \eqref{TraceCij} into
\begin{align}
\mathcal{C}_{tt} & = \frac{1}{2(2 \pi)^2} \left[ \left(\frac{1}{2} - 3 \xi\right) \Delta_h [w_\ell] + m^2 [w_\ell] - 2 [v_1] \right] - \frac{1}{8 \pi G_{\rm N}} \left[ \R - 4 \Lambda - 2 \alpha \Delta_h \R  - 6 \beta  \Delta_h \R \right] = 0.
\label{ConstraintTrAnom} 
\end{align}

Eq. \eqref{ConstraintTrAnom} implies the trace equation
\begin{align}
{\rm Tr} (G_{ab}) + 4 \Lambda = 8 \pi G_{\rm N} \omega({\rm Tr} (T_{ab}))
\end{align}

Setting $m = 0$ and $\xi = 1/6$ the constraint $\mathcal{C}_{tt} = 0$ \eqref{ConstraintTrAnom} is state-independent -- the only term that survives for the stress-energy tensor is the trace anomaly term -- and provides a criterion for the ultrastatic spacetime semiclassical solutions that admit conformally coupled fields.

If the field is not conformally coupled, then eq. \eqref{ConstraintTrAnom} is an elliptic equation for $[w_\ell]$, which can in principle be solved with an appropriate boundary condition. Inserting its solution into \eqref{Constraint3-DF} one then obtains a single set of constraints for the diagonals $[w_{\ell;ij}]$. It is important to emphasise, however, that $w_\ell$ is not enough to define the two-point function of a state, for the positivity condition is not guaranteed to hold.

%In this way, constraints \eqref{ConstraintTrAnom} and \eqref{Constraint3-DF} can be in principle solved explicitly, except in the conformally coupled case, where \eqref{ConstraintTrAnom} is purely geometric.

To summarise the discussion on the constraints of the theory in this section, we have that solutions to the semiclassical gravity equations must satisfy
\begin{subequations}
\begin{align}
\mathcal{C}_{tt} = 0, & \quad \quad \quad \mathcal{C}_{ij} = 0, \label{Constr-IVsurface} \\
\partial_t {\mathcal{C}_{tt}} = 0, & \quad \quad \quad \partial_t \mathcal{C}_{ij} = 0, \label{Constr-Conservation}
\end{align}
\end{subequations}
where $\mathcal{C}_{tt} = 0$ is defined by eq. \eqref{semiEFE-constraints1}, or more conviniently by eq. \eqref{ConstraintTrAnom}, and where $\mathcal{C}_{ij} = 0$ is defined by eq. \eqref{semiEFE-constraints3}, or more conviniently by eq. \eqref{Constraint3-DF}. The conservation of the constraints \eqref{Constr-Conservation} are automatically satisfied for time-translation invariant states that satisfy constraints \ref{Constr-IVsurface} and whose Wightman two-point function possesses the Hadamard singular structure. We shall study these states below in Sec. \ref{sec:States} below.

\section{States from Cauchy data in ultrastatic semiclassical gravity}
\label{sec:States}

We have seen above in Sec. \ref{subsec:Constraints} that the relevant states of semiclassical gravity in ultrastatic spacetimes are those that satisfy the constraints \eqref{Constr-IVsurface} on the initial value surface (at $t = 0$), and such that the constraints are preserved at all times, cf. \eqref{Constr-Conservation}. In this section, we study the class of states for which the conservation of constraints \eqref{Constr-Conservation} is guaranteed, which are time-translation invariant states. We shall assume that 
\begin{align}
A_h = -\Delta_h + m^2 + \xi R \text{ is positive in }L^2(S, \dd \vol_h),
\label{PositiveAh}
\end{align}
i.e., that for any $v \in L^2(S, \dd \vol_h)$, the map $v \mapsto (v, A_h v)_{L^2(S, \dd \vol_h)}$ is positive, and we give criteria for the Wightman function solutions only in this case.

%{\color{red} Use Kay Theorem 7.1 \cite{Kay:1978yp}}

We shall see that if \eqref{PositiveAh} holds, the initial data of constraint-preserving states can be determined only in terms of the datum $\omega(\varphi(x) \varphi(x'))$ (see \eqref{WightmanIVP}) since the condition of time-translation invariance will determine the form of the initial value surface correlation functions $\omega(\varphi(x) \pi(x'))$, $\omega(\pi(x) \varphi(x'))$ and $\omega(\pi(x) \pi(x'))$. % This is particularly convenient in view of the constraints $\mathcal{C}_{tt} = 0$ \eqref{ConstraintTrAnom} and $\mathcal{C}_{ij} = 0$ \eqref{Constraint3-DF}.

To this end, the first step is to obtain an explicit expression for the two-point function of the field in terms of initial data. Eq. (3.11) in \cite{Juarez-Aubry:2019jon} shows how to do so as distributional kernels % express the Wightman two-point function in terms of its initial data as a distribution 
with the aid of the causal propagator of the Klein-Gordon equation. The distributional kernel of the Wightman function in terms of initial-data kernels is given by eq. (3.19) in \cite{Juarez-Aubry:2019jon}. 

In our case of interest, i.e., if \eqref{PositiveAh} holds, then $A := \overline{A_h}$ is a positive, invertible, self-adjoint operator, in terms of which we can obtain the causal propagator, and hence the Wightman function in closed form. That $A = \overline{A_h}$ is positive, invertible and self-adjoint is a consequence of the following theorem by Kay:
\begin{thm}[Theorem 7.1 in \cite{Kay:1978yp}]
\label{Thm:Kay}
Let $(M = \mathbb{R} \times S, g)$ be ultrastatic and globally hyperbolic, and suppose that \eqref{PositiveAh} holds. Then $A_h$ is essentially self-adjoint on $C_0^\infty(S) \subset L^2(S, \dd \vol_h)$ and $A = \overline{A_h}$ is positive and invertible. \qed
\end{thm}

\begin{rem}
Theorem 7.1 in \cite{Kay:1978yp} actually requires that $S$ be geodesically complete, but this is an unnecessary requirement in the ultrastatic, globally hyperbolic case. See, e.g., theorem 1 on p. 130 of Fulling's monograph \cite{Fulling:1989nb}.
\end{rem}

We now proceed to obtain the causal propagator of the theory in terms of the functional calculus of $A$. We shall use def. 4 and 5 of \cite{Juarez-Aubry:2020aoo} for the causal propagator, which are particularly well-adapted to self-adjoint problems.

\begin{defn}[Def. 4 and 5 of \cite{Juarez-Aubry:2020aoo}]
\label{Def:CausalPropagator}
Let $(M = \mathbb{R}\times S,g)$ be an ultrastatic, globally hyperbolic spacetime. Let $\partial_t^2 + A_h$ be the Klein-Gordon operator with $A_h$ a symmetric, positive operator in $L^2(S, \dd\vol_h)$ defined on smooth functions of compact support, and let $A = \overline{A_h}$ be self-adjoint in $L^2(S, \dd\vol_h)$ with domain $D(A)$. We call an operator $\mathsf E:C^2_0(\mathbb R, D(A)) \to C^2(\mathbb R, D(A))$ \emph{an integral operator} (here $D(A)$ has the induced topology from $L^2(S, \dd\vol_h)$), and $E(t-t')$  the \emph{integral kernel of}  $ \mathsf E$, if 
\begin{equation}
\label{IntegralKernelDef}
(\mathsf Ef)(t,x)= \int_{\mathbb R} \dd t'\,( E(t-t')\, f(t'))(x),
\end{equation}
where $E(t), t \in \mathbb R,$ is a bounded operator on $L^2(S, \dd \vol_h)$. Furthermore,  we assume that  $E \in C^ 2(\mathbb R, \mathcal B(L^2(S, \dd \vol_h)).$       

$\mathsf E$ is a causal propagator of the Klein-Gordon operator $\partial_t^2 + A$ if it holds that
\begin{enumerate}[(i)]
\item 
\beq 
(\partial_t^2 + A) \mathsf E f =\mathsf  E (\partial_t^2 + A) f=0 , \qquad f \in C^2_0(\mathbb R, D(A)),
\ene
\item 
\beq
\text{\rm supp} \, \mathsf  (E f)_1 \subset  J ( {\rm supp}\,  f_1),  \qquad f \in C^2_0(\mathbb R, D(A)),
\ene
where ${\rm supp}\, g$ the support of $g.$

\item The integral kernel of $\mathsf E$ further satisfies,
\beq
E(0)=0, \qquad \partial_t E(0)= -I.
\ene
\end{enumerate}
\end{defn}

\begin{thm}
\label{Thm:CausalPropagator}
The integral operator $\mathsf{E}: C^2(\mathbb{R}, D(A)) \to C^2(\mathbb{R}, D(A))$ defined by the integral kernel
%\begin{subequations}
\begin{align}
%(\mathsf{E} f)(t,x) & := \int_\mathbb{R} \dd t' (E(t-t') f(t'))(x), \\
E(t-t') & := - \frac{\sin (A^{1/2}(t-t'))}{A^{1/2}}  
\label{CausalPropagator}
\end{align}
%\end{subequations}
is the causal propagator of $\partial_t^2 + A$.
\end{thm}
\begin{proof}
Items (i) and (iii) of def. \ref{Def:CausalPropagator} hold by the functional calculus of self-adjoint operators. Item (ii) holds by the support properties of the solutions of the Klein-Gordon equation in globally hyperbolic spacetimes \cite{Lichne}.
\end{proof}

Using theorem \eqref{Thm:CausalPropagator}, we can obtain explicitly the distributional kernel of the Wightman two-point function in terms of initial data in ultrastatic spacetimes. Further, we can characterise time-translation invariant states.

\begin{thm}
\label{Thm:InvariantWightman}
Let $G^+(t,t') = \omega(\Phi(t) \Phi(t'))$ be the distributional kernel of the Wightman two-point function of the Klein-Gordon field in the algebraic state $\omega$ in an ultrastatic, globally hyperbolic spacetime, $(M = \mathbb{R} \times S, g)$. Let
\begin{subequations}
\label{WightmanInitialData}
\begin{align}
G^+(0,0) = \omega(\varphi \varphi') = G^+_{\varphi \varphi}, & \quad & \partial_t G^+(0,0) = \omega(\pi \varphi') = G^+_{\pi \varphi}, \\
\partial_{t'}G^+(0,0) = \omega(\varphi \pi') = G^+_{\varphi \pi}, & \quad & \partial_t \partial_{t'} G^+(0,0) = \omega(\pi \pi') = G^+_{\pi \pi},
\end{align}
\end{subequations}
be the initial data of $G^+$ in terms of the correlation functions of the 3-field, $\varphi$, and its momentum, $\pi$, on the initial value surface $S$, defined at $t = 0$. Then
\begin{align}
G^+(t,t') & =  \cos(A^{1/2} t) \cos(A^{1/2} t') G^+_{\varphi \varphi} +  \frac{\cos(A^{1/2} t) \sin (A^{1/2} t')}{A^{1/2}} G_{\varphi \pi} \nonumber \\
& + \frac{\sin(A^{1/2} t) \cos( A^{1/2} t')}{A^{1/2}} G_{\pi \varphi} + \frac{\sin( A^{1/2} t) \sin (A^{1/2} t')}{A} G_{\pi \pi}.
\label{UltrastaticKGStates}
\end{align}

Moreover, time-translation invariant states satisfy the initial data constraints
\begin{subequations}
\label{StatInitialData}
\begin{align}
G^+_{\pi \pi} & = A G^+_{\varphi \varphi}, \label{StatInitialData1} \\
G^+_{\varphi \pi} & = \frac{\ii}{2} \1 = - G_{\pi \varphi}, \label{StatInitialData2}
\end{align}
\end{subequations}
whereby one has that time-translation invariant states take the form
\begin{align}
G^+(t,t') & = \cos(A^{1/2} (t-t')) G^+_{\varphi \varphi}  -  \ii \frac{\sin (A^{1/2} (t-t')) }{2A^{1/2}} =: \mathcal{G}^+(t-t'). \label{UltrastatInvWightman}
\end{align}
\end{thm}
\begin{proof}
Using eq. (3.19) in \cite{Juarez-Aubry:2019jon} and formula \eqref{CausalPropagator} for the causal propagator, one obtains eq. \eqref{UltrastaticKGStates}. Further, using functional calculus, we can express eq. \eqref{UltrastaticKGStates} as
\begin{align}
G^+(t,t') & = \frac{1}{2} \left(\cos(A^{1/2} (t-t')) + \cos(A^{1/2} (t+t'))  \right) G^+_{\varphi \varphi}  + \frac{-\sin (A^{1/2} (t-t')) + \sin (A^{1/2} (t+t'))}{2 A^{1/2}} G_{\varphi \pi} \nonumber \\
& + \frac{\sin (A^{1/2} (t-t')) + \sin (A^{1/2} (t+t'))}{2 A^{1/2}} G_{\pi \varphi} + \frac{\cos (A^{1/2} (t-t')) - \cos (A^{1/2} (t+t')) }{2 A} G_{\pi \pi}.
\label{UltrastatAuxWightman}
\end{align}

Time translation-invariant states have time dependence only appearing as the time difference $t-t'$. Imposing this on eq. \eqref{UltrastatAuxWightman}, together with the canonical commutation relations, $G^+_{\varphi \pi} - G^+_{\pi \varphi} = \ii \1$, yields eq. \eqref{StatInitialData}. (The second relation \eqref{StatInitialData2} satisfies the canonical commutation relations.)

In turn, if the initial data \eqref{WightmanInitialData} satisfies the relations \eqref{StatInitialData}, then eq. \eqref{UltrastatAuxWightman} takes the form of 
%\begin{align}
%G^+(t,s) & = G^+_{\varphi \varphi} \cos(A^{1/2} (t-t'))  -  \ii \frac{\sin (A^{1/2} (t-t')) }{2A^{1/2}},
%\end{align}
%which is 
eq. \eqref{UltrastatInvWightman}.
\end{proof}

\begin{rem}
Note that, as desired, we have the integral kernel identity $\Im \mathcal{G}^+(t-t') = \frac{1}{2} E(t-t')$ in eq. \eqref{UltrastatInvWightman}.
\end{rem}

If constraints \eqref{Constr-IVsurface} hold at $t = 0$ and the Wightman two-point function of the Klein-Gordon field is time-translation invariant, i.e. of the form of eq. \eqref{UltrastatInvWightman}, and the state is Hadamard, then the constraints \eqref{Constr-IVsurface} are satisfied at all times, since the limit $t'\to t$ is time independent. Thus, time-translation invariance is a necessary and sufficient condition for conditions \eqref{Constr-Conservation} to hold. %We shall henceforth focus on these states.

\subsection{Examples of time-translation invariant Hadamard states in ultrastatic spacetimes}

We now briefly give examples -- some quite well-known and important -- of time-translation invariant Hadamard states, and characterise them in terms of their initial data. Our purpose is merely to emphasise that this class of states contains well-known important examples

\subsubsection{The ground state}

The ground state is defined by the integral kernel ${}^\Omega \mathcal{G}^{+}(t-t') = \cfrac{\ee^{-\ii A^{1/2}(t-t')}}{2 A^{1/2}}$, and can be obtained from the initial data ${}^\Omega G^+_{\varphi \varphi} = \cfrac{1}{2 A^{1/2}}$, together with eq. \eqref{StatInitialData}. It is Hadamard by passivity \cite{Sahlmann:2000fh}. More explicitly, initial data can be expressed in terms of the complete set of eigenfunctions $\psi_j$ of $A$, such that $A \psi_j = \omega_j^2 \psi_j$, with the aid of the spectral theorem,
\begin{align}
{}^\Omega G^+_{\varphi \varphi}(x,x')  = \lim_{\epsilon \to 0^+} \int_{\sigma(A)} \dd \mu(j) \frac{\ee^{-\omega_j \epsilon}}{2 \omega_j} \psi_j(x) \overline{\psi_j(x')}.
\label{VaccumTwoPtSpectral}
\end{align} 

Note that if the spacetime is smooth, then the eigenfunctions $\psi_j$ are smooth by elliptic regularity, see e.g. \cite[p. 22]{Fulling:1989nb}.

\subsubsection{KMS states at positive temperature}

KMS states at positive temperature $1/\beta$ are defined by the integral kernel  
\begin{equation}
{}^{\beta} \mathcal{G}^{+}(t-t') = \cfrac{1}{2 A^{1/2} (1-\ee^{-\beta A^{1/2}})} \left( \ee^{-\ii A^{1/2}(t-t')} + \ee^{-\beta A^{1/2} + \ii A^{1/2}(t-t')}\right),
\end{equation} 
which can be obtained from the initial data ${}^\beta G^{+}_{\varphi \varphi} = \cfrac{1 + \ee^{-\beta A^{1/2}}}{2 A^{1/2} (1-\ee^{-\beta A^{1/2}})} = \cfrac{\coth\left(\frac{1}{2}\beta A^{1/2} \right)}{2 A^{1/2}}$, together with eq. \eqref{StatInitialData}. KMS states are Hadamard by passivity \cite{Sahlmann:2000fh}.

\subsubsection{``Smooth deviations" from vacuum and KMS states}

Consider the states defined by intial data
\begin{align}
{}^f G^+_{\varphi \varphi}(x,x')  = \lim_{\epsilon \to 0^+} \int_{\sigma(A)} \dd \mu(j) \left( \frac{\ee^{-\omega_j \epsilon}}{2 \omega_j} + f(\omega_j) \right) \psi_j(x) \overline{\psi_j(x')}.
\end{align} 
together with eq. \eqref{StatInitialData}, such that $f$ is bounded and $f \psi_j(x) \psi_j(x')$ decays rapidly (faster than any polynomial) in $\sigma(A)$. Then the integral
\begin{align}
\int_{\sigma(A)} \dd \mu(j) f(\omega_j) \psi_j(x) \overline{\psi_j(x')}
\end{align}
exists, and also arbitrarily many derivatives thereof, defining a smooth bi-function. The two-point function ${}^f G^+_{\varphi \varphi}(x,x')$ hence differs from the vacuum two-point function by a smooth piece, and has therefore Hadamard singular structure.

One can also form convex combinations of the above two-point functions of the form ${}^\alpha \mathcal{G}^{+}(t-t')= \alpha [{}^f \mathcal{G}^{+}(t-t')] + (1-\alpha) [{}^\beta \mathcal{G}^{+}(t-t')]$, where $\alpha \in [0,1]$. These are defined by the initial data induced from the data of the ground, KMS states and its ``smooth deviations". They are Hadamard by construction.

\subsection{Absence of coherent-state solutions of positive energy}
\label{subsec:Coherent}

Coherent states constitute another class of important states. They are quasi-free with non-vanishing one-point function. We show here that there are no non-trivial coherent states as solutions that are time-translation invariant. The reason is that the product of two dynamical, classical solutions will generally fail to be stationary.

Let us define the coherent state $\omega_\phi$ as ``peaked" around a classical solutions of the form 
\begin{equation}\label{ClassSol}
\phi(t) = \cos(t A^{1/2}) f + \frac{\sin (t A^{1/2})}{A^{1/2}} p.
\end{equation}
with initial data 
\begin{align}\label{ClassInitialData}
\left\{
                \begin{array}{l}
                  \phi(0,x) = f(x), \\
                  \partial_t \phi(0,x)  = p(x),
                \end{array}
\right.
\end{align}
with $f \in D(A)$ and $p \in D(A^{1/2})$.

We define $\omega_\phi$ in terms of its one-point function $\omega_\phi(\Phi(t)) = -\phi(t)$ and its two-point function $\omega_\phi(\Phi(t) \Phi(t')) =  {}^\phi G^{+}(t,t') = {}^\Omega \mathcal{G}^{+}(t-t') + \phi(t) \phi(t')$. The state will be time-translation invariant if the product $\phi(t) \phi(t')$ is. We have that
\begin{align}
\phi(t) \phi(t') & = \frac{1}{2} \left[\cos(t-t') + \cos(t+t') \right]f(t) f(t') + \frac{1}{2 A^{1/2}} \left[ \sin(t-t') + \sin(t+t') \right] f(t) p(t') \nonumber \\
& + \frac{1}{2 A^{1/2}} \left[\sin(t-t') + \sin(t+t') \right] p(t) f(t') + \frac{1}{2 A} \left[\cos(t-t') - \cos(t+t') \right] p(t) p(t').
\end{align}

Time-translation invariance requires that the following relations hold
\begin{align}
A^{1/2} f - p = 0, \quad \quad \quad A^{1/2} f + p = 0.
\end{align}
which imply that $A^{1/2} f = p = 0$. Inserting this into eq. \eqref{ClassSol}, we obtain that the time-translation invariant solution are the trivial solutions
\begin{equation}
\phi(t) = f,
\label{phif}
\end{equation}
where $f$ satisfies $A^{1/2} f = 0$. Looking at the form of the solutions \eqref{ClassSol}, in fact, it suffices that $f$ solve the elliptic equation $A f = 0$ and is time independent. But since our hypotheses require that $A$ be positive, $f = 0$. Thus, there are no coherent-state solutions of positive energy.

%Thus, we have that time-translation invariant states have a Wightman two-point function integral kernel
%\begin{equation}
%\omega_\phi(\Phi(t) \Phi(t')) =  {}^\phi \mathcal{G}^{+}(t-t') = {}^\Omega \mathcal{G}^{+}(t-t') + f f',
%\label{TrivialCoherent}
%\end{equation} 
%where $A f = 0$. % characterised by initial data ${}^\phi G^{+}_{\varphi\varphi}(t-t') = {}^\Omega G^{+}_{\varphi\varphi}(t-t') + f f'$, together with eq. \eqref{StatInitialData}.

%The Hadamard condition is guaranteed by elliptic regularity. Namely, if the 3-metric $h$ on $S$ is smooth, then the elliptic operator $A$ has smooth coefficients, which implies that $f$ in eq. \eqref{phif} is smooth, see e.g. \cite[p. 22]{Fulling:1989nb}, from where it follows that time-translation invariant coherent states are Hadamard for smooth 3-metric $h$. Note, however, that the last term on the right-hand side of eq. is not dynamical. In this sense, there are no non-trivial coherent states that are time-translation invariant.

\subsection{Remarks on states as solutions in ultrastatic semiclassical gravity}

We have seen thus far that the states that solve semiclassical gravity in ultrastatic spacetimes must be time-translation invariant, and of the form of eq. \eqref{UltrastatInvWightman}. We must emphasise, however, that not every state of this form is a solution of semiclassical gravity in ultrastatic spacetimes. In particular, it is required that the constraints \eqref{Constr-IVsurface} hold on the initial surface.

Whether such constraints may hold depends on the details of the problem, since the constraints depend strongly on the geometry coefficients $\Lambda$, $G_{\rm N}$, $\alpha$ and $\beta$, the  as well as on the field parameters $m^2$ and $\xi$. Even the existence of a vacuum state for the solutions is not guaranteed in general.

If the Cauchy surface $S$ is non-compact, then the singularities of the integrand can be isolated and subtracted by power-counting in the spectral parameter. The coincidence limit can be then examined for the regular part, and inserted in eq. \eqref{ConstraintTrAnom} as a first criterion. Afterwards, constraint \eqref{Constraint3-DF} can be analysed. 

If the Cauchy surface $S$ is closed (compact without boundary), then large eigenvalue estimates -- Weyl's law for the operator $A$ -- can best isolate the singular structure. We should emphasise, however, that this task can be very technical in its own right. In situations with large amount of symmetries, other symmetry-adapted arguments may suffice, such as reading off expansion coefficients. We study below some relatively simple examples in sec. \ref{sec:Examples}.

We should add a word on the r\^ole of normal ordering in the ultrastatic context. Distinguished annihilation and creation operators exist since the ``positive energy" projector in the one-particle structure of the quantum field theory is distinguished by the time-translation symmetry. (Equivalently, a distinguished complex structure can be selected for complexified classical solutions.) Choosing the $w_\ell^0$ coefficient for $H_\ell$, cf. eq. \eqref{Hadamard}, in such a way that it coincides with the $w_\ell$ coefficient of the vacuum state modulo terms $O(\sigma^{3/2})$ is tantamount to subtracting the vacuum two-point function from the state in order to remove the singular structure of the two-point function. This procedure is equivalent to a normal ordering prescription. 

In the context of normal ordering, $[w_\ell]$ and $[w_{\ell;ij}]$ are traded to $[w_\ell- w_{\ell, {\rm vac}}]$ and $[w_{\ell;ij}- w_{\ell, {\rm vac};ij}]$ in eq. \eqref{Constraint3-DF} and \eqref{ConstraintTrAnom}. If one takes the normal ordering prescription viewpoint, the existence of the vacuum state is a purely geometric constraint: spacetimes that admit vacuum states are vacuum (matterless) solutions for higher-order gravity, with the higher order terms stemming from $[v_1]$, $I_{ab}$ and $J_{ab}$.

%{\color{red} Make rigorous argument that can use normal ordering. Two cases: $S$ non-compact, in which case $\sigma(A)$ is continuous, and use normal ordering. $S$ compact, in which case Hadamard subtraction requires knowledge of Weyl's laws for large eigenvalues.}

\section{Semiclassical gravity in static spacetimes}
\label{sec:Stat}

We now come to the task of studying semiclassical gravity in uniformly static, globally hyperbolic spacetimes. It is to our advantage that static spacetimes are conformally related to ultrastatic ones by a time-independent conformal factor. Namely, if $\tilde{g}_{ab}$ is static, it is true that
\begin{align}
\tilde g_{ab}&  = \Theta^2 g_{ab}, \hspace{1cm} \Theta := \ee^\theta > 0, \hspace{1cm} g_{ab} = - \dd t_a \otimes \dd t_b + h_{ij} \dd x^i_a \otimes \dd x^j_b,
\end{align}
where $0 < c_1\leq \Theta< c_2$ is time-independent, with $c_1$ and $c_2$ real constants.

Using the conformal transformation laws for the curvature tensors (see e.g. App. D of \cite{WaldBook}) we have that,
\begin{subequations}
\label{ConformalCurvature}
\begin{align}
\tilde R_{abc}{}^d & = R_{abc}{}^d + 2 \delta^d{}_{[a} \nabla_{b]} \nabla_c \theta - 2 g^{de} g_{c[a}\nabla_{b]}\nabla_e \theta + 2(\nabla_{[a} \theta) \delta^d{}_{b]} \nabla_c \theta - 2 (\nabla_{[a} \theta) g_{b]c} g^{df} \nabla_f \theta - 2 g_{c[a} \delta^d{}_{b]} g^{ef} (\nabla_e \theta ) \nabla_f \theta  ,\\
\tilde R_{ac} & = R_{ac}  - 2 \nabla_a \nabla_c \theta - g_{ac} g^{de} \nabla_d \nabla_e \theta + 2 (\nabla_a \theta) \nabla_c \theta - 4 g_{ac} g^{de} (\nabla_d \theta) \nabla_e \theta,\\
\tilde R & = \Theta^{-2} \left[R - 6 g^{ac}\nabla_a \nabla_c \theta - 6 g^{ac} (\nabla_a \theta) \nabla_c \theta \right],
\end{align}
\end{subequations}
where covariant derivatives are with respect to the 3-metric $h$, and tildes denote the curvature tensors with respect to the static metric $\tilde{g}$. As is well known, the field theory of ultrastatic spacetimes can also be exported to the static case with the aid of the so-called \emph{optical metric}, which is precisely the (globally hyperbolic) ultrastatic metric $g$. The Klein-Gordon equation in static spacetimes,
\begin{align}
\left( \tilde \Box - m^2 - \xi \tilde R \right) \tilde \Phi = 0,
\label{KG-Static}
\end{align} 
can be cast in the form (see e.g. Chap. 6 of \cite{Fulling:1989nb})
\begin{align}
 \left[ \partial_t^2 - \Delta_h + \Theta^2 \left( m^2 + \left( \xi - \frac{1}{6}\right) \tilde R\right) + \frac{1}{6} R \right] \Phi = 0, \hspace{1cm} & \Phi := \Theta \tilde \Phi.
 \label{KG-Conformal}
\end{align}

Using conformal techniques, we can construct the time-translation invariant states in the static case. Assume that
\begin{align}
\tilde{A}_h := - \Delta_h + \Theta^2 \left( m^2 + \left( \xi - \frac{1}{6}\right) \tilde R\right) + \frac{1}{6} R > 0 \quad \text{ in } \quad L^2(S, \dd \vol_h),
\label{Atilde}
\end{align}
then by theorem \eqref{Thm:Kay} $\tilde A := \overline{\tilde{A}_h}$ is self-adjoint in $L^2(S, \dd \vol_h)$ and ``auxiliary" Wightman two-point functions for $\Phi$ can be constructed in terms of initial data, which will be a bi-solutions to eq. \eqref{KG-Conformal}. In terms of these auxiliary Wightman function, it follows that
\begin{align}
\tilde G^+(\x, \x') = \Theta^{-1}(\x) G^+(\x, \x') \Theta^{-1}(\x')
\label{ConformalWightman}
\end{align}
is a bi-solution to eq. \eqref{KG-Static} and \eqref{KG-Conformal}, and Wightman function for the Klein-Gordon field $\tilde \Phi$.

Using theorem \ref{Thm:InvariantWightman}, time-translation invariant states in static spacetimes take the form
\begin{subequations}
\label{StaticWightman}
\begin{align}
\tilde G^+(t, t') & =  \Theta^{-1} \cos(\tilde A^{1/2}(t-t')) G^+_{\varphi \varphi} (\Theta')^{-1}- \ii \Theta^{-1} \frac{\sin(\tilde{A}^{1/2}(t-t'))}{2 \tilde{A}^{1/2}} (\Theta')^{-1} =: \tilde{\mathscr{G}}(t-t'), \\
\tilde G^+_{\varphi \varphi} & := \Theta^{-1} G^+_{\varphi \varphi} (\Theta')^{-1}.
\end{align}
\end{subequations}

\begin{prop}
Let $G^+$ be a Hadamard Wightman function in $(M, g_{ab})$ and a bi-solution to eq. \eqref{KG-Conformal}, then the Wightman function defined in eq. \eqref{ConformalWightman} as a bi-solution to eq. \eqref{KG-Static} is Hadamard in $(M, \tilde{g}_{ab})$.
\end{prop}
\begin{proof}
Since $\Theta$ is smooth and strictly positive, as bidistributions ${\rm WF}(G^+) = {\rm WF} (\tilde{G}^+)$. Moreover, the antisymmetric part of eq. \eqref{ConformalWightman} defines the causal propagator of $\tilde \Box - m^2 - \xi \tilde{R}$ via the kernel
\begin{align}
\tilde{E}(\x, \x') := \Theta^{-1}(\x) E(\x, \x') \Theta^{-1}(\x')
\end{align}
where ${\sf E}$ is the causal propagator of $(\partial_t^2 - \tilde{A}_h)$ with kernel $E$. To see this, note that by a direct calculation
\begin{align}
\left(\tilde \Box - m^2 - \xi \tilde{R}\right) \int_M \dd \vol(\x') \tilde{E}(\x, \x') f(\x') & = \Theta^{-3}(\x) (\partial_t^2 - \tilde{A}_h) \Theta(\x) \int_M \dd \vol(\x') \tilde{E}(\x, \x') f(\x')  \nonumber \\
& = \Theta^{-3} (\partial_t^2 - \tilde{A}_h) {\sf E} (\Theta^3 f) =  (\partial_t^2 - \tilde{A}_h) {\sf E} f = 0,
\end{align}
and that
\begin{align}
 \int_M \dd \vol(\x') \tilde{E}(\x, \x') \left(\tilde \Box' - m^2 - \xi \tilde{R}(\x')\right) f(\x') & =  \Theta^{-1}(\x)\int_M \dd \vol(\x') E(\x, \x') \Theta^{-4}(\x') (\partial_t^2 - \tilde{A}_h') \Theta(\x') f(\x') \nonumber \\
 & = \Theta^{-1} {\sf E} (\partial_t^2 - \tilde{A}_h) (\Theta f) = {\sf E} (\partial_t^2 - \tilde{A}_h) f = 0.
\end{align}

(Analogous calculations can be used to prove similar relations for the advanced and retarded propagators.) Furthermore, the distributional support properties are unchanged under the conformal rescaling, i.e., $\supp \tilde{E} = \supp E$.

Using Radzikowski's theorem \cite[Theorem 5.1]{Radzikowski:1996pa}, in particular the equivalence of items 1 and 3 in that theorem, the result follows.

\end{proof}

It follows that the Hadamard singular structure can be read off immediately from the above expression
\begin{align}
\tilde{H}_\ell(\x, \x') = \Theta^{-1}(\x) H_\ell(\x, \x') \Theta^{-1}(\x'),
\label{StaticHadamard}
\end{align}
where $H_\ell$ is the Hadamard singular structure for the auxiliary ultrastatic problem \eqref{KG-Conformal}.

%{\color{red} EDITS BEGIN HERE}

%\begin{align}
%\tilde \sigma = \Theta^2 \sigma - \frac{\Theta}{2} \left(g_{ab} \nabla_c \Theta + g_{ca} \nabla_b \Theta + g_{bc} \nabla_a \Theta  \right) \sigma^{;a} \sigma^{;b} \sigma^{;c}
%\end{align}

%{\color{red} EDITS END HERE}

The remaining element to obtain the initial-data constraints stemming from the semiclassical gravity in the static case explicitly is, in view of eq. \eqref{TabPointSplit2} and \eqref{TabSummary}, to obtain the pararallel-transport propagator, $\tilde{g}^\mu{}_{\mu'}$. The spacetime tetrads of $\tilde{g}_{ab}$ can be obtained by conformally rescaling those of $g_{ab}$. Indeed, $\tilde{e}^I_\mu = \Theta e_\mu^I$ and $\tilde{e}^\mu_I = \Theta^{-1} e_I^\mu$, from where it follows that
\begin{align}
\tilde{g}^{\mu}{}_{\mu'}(p,p') = \Theta^{-1}(p) g^{\mu}{}_{\mu'}(p,p') \Theta(p').
\label{StatParallelProp}
\end{align}

%Using the covariant differentiation rule,
%\begin{align}
%(\tilde{\nabla}_a - \nabla_a)\omega_b = \left[ g_{ab} \nabla^c \theta -  2 \delta^c{}_{(a}\nabla_{b)} \theta \right] \omega_c
%\end{align}
%together 
From eq. \eqref{StaticWightman}, \eqref{StaticHadamard}, \eqref{StatParallelProp} and \eqref{ConformalCurvature} we can obtain the expectation value of the renormalised stress-energy tensor in static situations in terms of the geometry of the 3-manifold geometry $(S, h_{ab})$ and the conformal factor $\Theta$, and write down an explicit expression for the constrained initial value problem \eqref{semiEFEKG2}. We have that initial data must satisfy the following constraints,
%\begin{subequations}
%\label{semiEFEStat}
%\begin{align}
%&  \frac{1}{2} \R - \Lambda = 8 \pi G_{\rm N} \lim_{p' \to p} \mathcal{T}_{tt} \left(G^+(p, p') - H_\ell (p, p') \right) - \frac{G_{\rm N}}{\pi}  [v_1] + \frac{\alpha}{2} \left(\Delta_h \R + \R^{ij}\R_{ij}\right) + \beta \left(2 \Delta_h \R + \frac{1}{2} \R^2 \right), \label{semiEFEStat-1}\\
%& \lim_{p' \to p} \mathcal{T}_{t i} \left(G^+(p, p') - H_\ell (p, p') \right) = 0 \label{semiEFEStat-2}\\
%& \R_{ij} - \frac{1}{2} \R h_{ij} + \Lambda h_{ij} = 8 \pi G_{\rm N} \lim_{p' \to p} \mathcal{T}_{ij} \left(G^+(p, p') - H_\ell (p, p') \right) + \frac{G_{\rm N}}{\pi} h_{ij} [v_1] + \alpha I_{ij} + \beta J_{ij}, \label{semiEFEStat-3}\\
%& (\partial_t^2 + \tilde{A}_h) \Theta \tilde{G}^+ = (\partial_{t'}^2 + \tilde{A}_{h'}) \Theta' \tilde{G}^+ = 0.  \label{KGStat}
%\end{align}
%\end{subequations}

%Eq. \eqref{semiEFEStat-1}, \eqref{semiEFEStat-2} and \eqref{semiEFEStat-3} can be cast in the form

\begin{subequations}
\label{ConstrStat}
\begin{align}
\tilde{\mathcal{C}}_{tt} & = \frac{1}{8 \pi^2}\left[-[\tilde \nabla_t^2 \tilde{w}_{\ell}] - \frac{1}{2}\left(2 \xi - \frac{1}{2} \right)  h^{ij} \tilde{\nabla}_i \tilde{\nabla}_j [\tilde{w}_\ell] + \xi \tilde{R}_{tt} [\tilde{w}_\ell] + \Theta^2  [\tilde{v}_1]   \right] \nonumber \\
&  - \frac{1}{8 \pi G_{\rm N}} \left[\tilde{R}_{tt} + \Theta^2 \left( \frac{1}{2} \tilde{R} - \Lambda \right) - \alpha \tilde{I}_{tt} - \beta \tilde{J}_{tt} \right] = 0, \label{CttStat}\\
\tilde{\mathcal{C}}_{ti} & =  -\frac{1}{8 \pi^2} [\tilde{\nabla}_i \tilde{\nabla}_t\tilde{w}_{\ell}] + \frac{1}{8 \pi G_{\rm N}} \left[ \alpha \tilde{I}_{ti} + \beta \tilde{J}_{ti} \right] = 0  , \label{CtiStat} \\
\tilde{\mathcal{C}}_{ij} & = \frac{1}{8 \pi^2}\left[-[\tilde{\nabla}_i \tilde{\nabla}_j \tilde{w}_{\ell}] + \frac{1}{2}(1-2\xi) \tilde{\nabla}_j \tilde{\nabla}_i [\tilde{w}_\ell] + \frac{1}{2}\left(2 \xi - \frac{1}{2} \right) h_{ij}  h^{kl} \tilde{\nabla}_k \tilde{\nabla}_l [\tilde{w}_\ell] + \xi \tilde{R}_{ij} [\tilde{w}_\ell] - \Theta^2 h_{ij} [\tilde{v}_1]   \right] \nonumber \\
& - \frac{1}{8 \pi G_{\rm N}} \left[\tilde{R}_{ij} - \Theta^2 \left( \frac{1}{2} \tilde{R} - \Lambda \right)  h_{ij} - \alpha \tilde{I}_{ij} - \beta \tilde{J}_{ij} \right] = 0  \label{CijStat}.
\end{align}
\end{subequations}

All covariant derivatives in eq. \eqref{ConstrStat} are with respect to the metric $\tilde{g}$, as indicated by the tildes, but can be cast as covariant derivatives with respect to $g$ by the relation
\begin{align}
(\tilde{\nabla}_a - \nabla_a)v_b = \left[ g_{ab} \nabla^c \theta -  2 \delta^c{}_{(a}\nabla_{b)} \theta \right] v_c.
\end{align}

Let us now study constraints \eqref{ConstrStat} in more detail. First, by an argument analogous to the one in the ultrastatic case for the constraint $\mathcal{C}_{ti} = 0$, cf. Lemma \ref{LemmaCti}, we have that in eq. \eqref{CtiStat} the term $[w_{\ell;ti}] = 0$, thus, only the geometric terms of constraint \eqref{CtiStat} remain. Moreover, using eq. \eqref{ConformalCurvature}, we can see that the geometric terms vanish,
\begin{align}
8 \pi G_{\rm N} \tilde{\mathcal{C}}_{ti} & =   \alpha \left( - \tilde \Box \tilde{R}_{it} + 2 \tilde{R}_i{}^\mu \tilde{R}_{\mu t} \right) - 2 \beta  \tilde R \tilde{R}_{ti} = 2 \alpha \tilde{R}_i{}^t \tilde{R}_{tt}  = 2 \alpha \tilde{g}^{t \mu}\tilde{R}_{i\mu} \tilde{R}_{tt} = 2 \alpha \tilde{g}^{t j}\tilde{R}_{ij} \tilde{R}_{tt} = 0. \label{CtiStatVanish}
\end{align}

Thus, again in the static case, $\tilde{\mathcal{C}}_{ti} = 0$ is an identity and the stress-energy tensor is block-diagonal.

Further, taking the trace of eq. \eqref{CijStat} with respect to the 3-metric $h$ and using eq. (52) in \cite{Decanini:2005eg}, which yields the relation
\begin{align}
-[\tilde{\nabla}_t^2 \tilde{w}_{\ell}] =  - h^{ij} [\tilde{\nabla}_j \tilde{\nabla}_i\tilde{w}_{\ell}] + (m^2 + \xi \tilde R)\Theta^2 [\tilde{w}_\ell] -6 \Theta^2 [\tilde{v}_1],
\end{align}
one can recast constraint \eqref{CttStat} as an elliptic equation for $[\tilde{\omega}_\ell]$, as in the ultrastatic case,
\begin{align}
\tilde{\mathcal{C}}_{tt} & = \frac{1}{8 \pi^2}\left[\left( \frac{1}{2} - 3 \xi \right)  h^{ij} \tilde{\nabla}_i \tilde{\nabla}_j [\tilde{w}_\ell]   + m^2 \Theta^2 [\tilde{w}_\ell] -2 \Theta^2 [\tilde{v}_1]  \right]  - \frac{\Theta^2}{8 \pi G_{\rm N}} \left[ \tilde{R} -4 \Lambda + \alpha  \tilde{g}^{\mu \nu}\tilde{I}_{\mu \nu} + \beta \tilde{g}^{\mu \nu} \tilde{J}_{\mu \nu}  \right] = 0.
\label{Ctt-Stat-Trace}
\end{align}

Constraint \eqref{Ctt-Stat-Trace} is the trace of the semiclassical Einstein equations. In the conformally coupled case, eq. \eqref{Ctt-Stat-Trace} becomes state independent and the constraint becomes purely geometric in terms of the trace anomaly of the renormalised stress-energy tensor.

%\begin{align}
%h^{ij}\tilde{\mathcal{C}}_{ij} & = \frac{1}{8 \pi^2}\left[-h^{ij}[\tilde{\nabla}_i \tilde{\nabla}_j \tilde{w}_{\ell}] + \frac{\Theta^2}{2} \tilde{\Box} [\tilde{w}_\ell] + \frac{3 \Theta^2}{2}\left(2 \xi - \frac{1}{2} \right) \tilde \Box [\tilde{w}_\ell] + \xi h^{ij} \tilde{R}_{ij} [\tilde{w}_\ell] - 3 \Theta^2 [\tilde{v}_1]   \right] \nonumber \\
%& {\color{red} - \frac{1}{8 \pi G_{\rm N}} \left[h^{ij}\tilde{R}_{ij} - 3 \Theta^2 \left( \frac{1}{2} \tilde{R} - \Lambda \right) - \alpha h^{ij}\tilde{I}_{ij} - \beta h^{ij} \tilde{J}_{ij} \right]} = 0
%\end{align}

\section{Examples and no-go results}
\label{sec:Examples}

As we have mentioned before, solving the constraints explicitly in particular examples is an involved issue. In this section, we discuss some relatively simple examples and a no-go result for conformally-coupled fields. 

\subsection{Ultrastatic spacetimes with maximally-symmetric spatial geometry}
\label{Maximally-symmetric-ultrastatic}

Maximally symmetric Riemannian geometries in three dimensions are (i) Euclidean geometry, (ii) elliptic geometry and (iii) hyperbolic geometry. Case (i) is Minkowski spacetime, where it is well known that an appropriate choice of renormalisation parameters ensures that the renormalised stress-energy tensor in the Minkowski vacuum vanishes, in agreement with Wald's fourth stress-energy renormalisation axiom \cite{Wald:1978pj}, and thus solving the semiclassical gravity equations globally. Case (ii) has been studied by Sanders in great detail in \cite{Sanders:2020osl}. 

Case (iii) requires an involved analysis -- indeed analogous to that of case (ii) --, which requires a detailed treatement in its own right. Here we shall study some special cases for which can show the existence of solutions to semiclassical gravity, and comment briefly on the setup of the more general cases, which deserve attention in its own right elsewhere. 

The spatial section in case (iii) is the maximally symmetric Riemannian space $(S, h_{ab})$ with
\begin{align}
h_{ab} =  \rho^2 \dd r_a \otimes \dd r_b  +\rho^2 \sinh^2 r \dd \theta_a \otimes \dd \theta_b  + \rho^2  \sinh^2 r \sin^2 \theta \dd \phi_a \otimes \dd \phi_b,
\end{align}
where $r \in \mathbb{R}^+$, $\theta \in [0, \pi]$ and $\phi \in [0, 2\pi)$, and $\rho>0$ is the hyperbolic radius.
The curvature tensors and scalars of $(S, h_{ab})$ are given by
\begin{subequations}
\begin{align}
\R_{abcd}  = - \frac{1}{\rho^2} (h_{ac} h_{bd} - h_{ad} h_{bc} ), \hspace{50pt} \R_{ab}  = - \frac{2}{\rho^2} h_{ab},  \hspace{50pt} \R  = - \frac{6}{\rho^2}.
\end{align}
\end{subequations}

In the general case the relevant eigenvalue problem at hand is $A_h \Psi = \omega^2 \Psi$ with the operator
\begin{align}
A_h := - \Delta_h + m^2 - \frac{6 \xi}{\rho^2} = -\frac{1}{\rho^2} \left(\partial_r^2 + 2 \coth r \partial_r + \frac{1}{\sinh^2 r} \Delta_{\mathbb{S}^2} - \rho^2 m^2 + 6 \xi \right).
\end{align}

Under the assumption $m^2 - 6 \xi/\rho \geq 0$, we have that $A_h$ is positive, and we fall in the hypotheses of section \ref{sec:States}. The central point for the construction of the states is finding an explicit form of the eigenfunctions and the spectrum of (the self-adjoint extension of) $A_h$. The angular part is solved by spherical harmonics, and one is left with the radial equation 
\begin{align}
- \left(\partial_r^2 + 2 \coth r \partial_r + \frac{l(l+1)}{\sinh^2 r}  \right) \psi_{l \omega} = - \frac{1}{\sinh^2 r} \left( \partial_r \left( \sinh^2 r \partial_r \right) + l(l+1) \right) \psi_{l \omega} = \left( \rho^2 \omega^2 - \rho^2 m^2 + 6 \xi \right)\psi_{l \omega}.
\label{Radial}
\end{align}

Eq. \eqref{Radial} defines a singular Sturm-Liouville problem, and the strategy for constructing physical states amounts to obtaining the resolvent operator as an eigenfunction expansion of the Sturm-Liouville differential operator of the problem. (The differential operator in the middle of eq. \eqref{Radial} has explicit Sturm-Liouville form.) Eq. \eqref{Radial} admits two linearly independent solutions in terms of hypergeometric functions. The technical point is to appropriately normalise the radial eigenfunctions $\psi_{l \omega}$ as a linear combination of these linearly independent solutions by imposing a resolution of the identity as a delta-function eigenfunction expansion. Once this is achieved, the states for the theory can be expressed as eigenfunction expansions with eigenfunctions $\Psi_{\omega l m}(r, \theta, \phi) = \psi_{l \omega}(r) Y_{lm}(\theta, \phi)$, and one is left to verify the constraints \eqref{Constraint3-DF} and \eqref{ConstraintTrAnom}. A detailed analysis of this problem deserves attention in its own right.

In the special case in which $m^2 + \xi \R = 1$ the problem becomes simplified \cite{Sanders:2020osl}. We have for the vacuum that
\begin{align}
G^+_{\mathbb{H}}(\x, \x') = \frac{1}{4 (2\pi)^2 \rho^2} \frac{1}{\sinh \left( \Delta^+(\x, \x') \right) \sinh \left( \Delta^-(\x, \x') \right) }
\label{G+Hyperbolic}
\end{align}
where $\Delta^\pm(\x, \x') := \rho^{-1}(\pm(t-t')/2 + \mu(\underline{\x}, \underline{\x}')/2)$, and $\mu$ is the geodesic distance in hyperbolic space, so that $\sigma_h = \mu^2/2$. The singular structure of $G^+$ \eqref{G+Hyperbolic} is
\begin{align}
H(\x, \x') = \frac{1}{4 (2  \pi)^2 \rho^2 \Delta^+(\x, \x') \Delta^-(\x, \x')} = \frac{1}{2(2 \pi)^2 \sigma(\x, \x')}.
\end{align}

All that we need to do in order to obtain semiclassical gravity solutions is to extract the datum
\begin{align}
G^+_{\mathbb{H}\varphi \varphi} = G^+_{\mathbb{H}}(\x, \x')|_\mathcal{C}
\end{align}
and verify whether the $\mathcal{C}_{ij}$ \eqref{Constraint3-DF} and $\mathcal{C}_{tt}$ \eqref{ConstraintTrAnom} constraints can be imposed -- which, as we shall see, they can. Using the fact that $\Delta^+|_\mathcal{C} = \Delta^-|_\mathcal{C} = \mu/2 \rho=:\Delta$, it follows that
\begin{align}
G^+_{\mathbb{H}\varphi \varphi}-H|_\mathcal{C} = \frac{1}{4(2 \pi)^2 \rho^2 \Delta^2} \left(- \frac{\Delta^2}{3} + \frac{\Delta^4}{15} + O(\Delta^6) \right) = \frac{1}{12(2 \pi)^2 \rho^2} \left(-1 + \frac{\sigma_h}{10 \rho^2} + O(\sigma_h^2) \right), 
\end{align}
and we have that
\begin{subequations}
\label{w-HyperbolicSpace}
\begin{align}
[w] & = -\frac{1}{12(2 \pi)^2 \rho^2}, \\
[w_{;ij}] & = \frac{1}{120(2 \pi)^2 \rho^4} h_{ij},
\end{align}
\end{subequations}
which allows us to analyse the $\mathcal{C}_{ij}$ \eqref{Constraint3-DF} and $\mathcal{C}_{tt}$ \eqref{ConstraintTrAnom} constraints. The $\mathcal{C}_{tt}$ constraint fixes the cosmological constant as
\begin{align}
\frac{4 \pi}{ G_{\rm N}}  \Lambda    = \frac{m^2}{12(2 \pi)^2 \rho^2} + 2 [v_1] + \frac{\pi}{ G_{\rm N}}  \R = \frac{m^2}{12(2 \pi)^2 \rho^2} + 2 [v_1] - \frac{6 \pi}{ G_{\rm N} \rho^2},
\end{align}
where $[v_1]$ takes the constant value in terms of $m^2$, $\xi$ and $\rho$,
\begin{align}
 [v_1] & =   \frac{1}{8} m^4 - \frac{3}{2} \left(\xi - \frac{1}{6}\right) \frac{m^2}{\rho^2}   + \frac{9}{2} \left(\xi - \frac{1}{6} \right)^2 \frac{1}{\rho^4}.
\end{align}

The $\mathcal{C}_{ij}$ constraint reads
\begin{align}
 -[w_{;ij}] - \frac{2\xi}{\rho^2} [w] h_{ij}   = [v_1]  h_{ij} + \frac{\pi}{ G_{\rm N}} \left[ \frac{1}{\rho^2}  + \Lambda + \frac{2 \alpha }{\rho^4} + \frac{6 \beta }{\rho^4} \right] h_{ij}.
\label{wijHyperbolicSpace}
\end{align}

Inserting eq. \eqref{w-HyperbolicSpace} into \eqref{wijHyperbolicSpace} we obtain a relation for the combination $\alpha + 3 \beta$ of the renormalisation ambiguities, given by
\begin{align}
\frac{2\pi(\alpha + 3 \beta)}{ G_{\rm N} \rho^4} =  -\frac{1}{120(2 \pi)^2 \rho^4} +  \frac{\xi}{6(2 \pi)^2 \rho^4}  -[v_1] - \frac{\pi}{ G_{\rm N}} \left( \frac{1}{\rho^2}  + \Lambda  \right) .
\end{align}

Thus, we can see that there is a large number of semiclassical gravity solutions with hyperbolic spatial section, by appropriately choosing the renormalisation ambiguities.

% CALCULATIONS:

%\begin{align}
% [v_1] & =   \frac{1}{8} m^4 - \frac{3}{2} \left(\xi - \frac{1}{6}\right) \frac{m^2}{\rho^2}   + \frac{9}{2} \left(\xi - \frac{1}{6} \right)^2 \frac{1}{\rho^4} ,  \\
%I_{ij} & = - \frac{8}{\rho^4} h_{ij}  + \frac{6}{\rho^4} h_{ij} = - \frac{2}{\rho^4} h_{ij} , \\
%J_{ij} & =  \frac{18}{\rho^4}h_{ij} - \frac{24}{\rho^4}  h_{ij} = -\frac{6}{\rho^4} h_{ij}. 
%\end{align}

%\begin{align}
%\R^{ij}\R_{ij} & = \frac{12}{\rho^4} \nonumber \\
%\R_{ijkl} \R^{ijkl} & =   \frac{1}{\rho^4} (h_{ik} h_{jl} - h_{il} h_{jk} ) (h^{ik} h^{jl} - h^{il} h^{jk} ) = \frac{1}{\rho^4} (18 - h_{ik} h_{jl} h^{il} h^{jk} - h_{il} h_{jk} h^{ik} h^{jl}) = \frac{1}{\rho^4}(18 - 6) = \frac{12}{\rho^4} \\
%\R^{kl} \R_{kilj} & =  \frac{2}{\rho^4} h^{kl}(h_{kl} h_{ij} - h_{kj} h_{il} ) = \frac{4}{\rho^4} h_{ij}
%\end{align}

\subsection{Ricci-flat, static spacetimes with conformally-coupled fields}
\label{RicciFlat-static}

Ricci-flat spacetimes are particularly interesting examples, being vacuum solutions in General Relativity, which necessarily must have a vanishing stress-energy tensor. The situation is different in semiclassical gravity due to the presence of the Riemann tensor in the $[v_1]$ term, contributing geometrically even if the Wightman function of the quantum matter fields yields non-trivial contributions to the stress-energy tensor. Here, we present a no-go result for conformally-coupled fields, which states that in this case there are no semiclassical gravity solutions with non-trivial curvature.

%\subsection{Non-existence of solutions for non-maximally-symmetric, static, Ricci-flat spacetimes with conformally-coupled fields}
%\label{Non-existence}

\begin{prop}
\label{Prop-no-go}
Let $(M, g_{ab})$ be a static, Ricci-flat, globally hyperbolic spacetime. If $(M, g_{ab})$ is not Minkowski spacetime, there exist no states for the conformally coupled Klein-Gordon field ($m^2 = 0$, $\xi = 1/6$) yielding semiclassical gravity solutions for any value of the renormalisation ambiguity coefficients $\alpha$ and $\beta$.
\end{prop}
\begin{proof}
Imposing constraint \eqref{Ctt-Stat-Trace} yields
\begin{align}
\frac{1}{720} R_{\mu \nu \rho \sigma} R^{\mu \nu \rho \sigma} = [v_1] = \frac{2 \pi \Lambda }{G_{\rm N}},
\label{MaxSymmetric}
\end{align}
which implies that $g_{ab}$ must be maximally symmetric. The only Ricci-flat, maximally symmetric spacetime is Minkowski spacetime with $\Lambda = 0$.
\end{proof}

The above proposition implies, for example, that there exist no semiclassical solutions for the Schwarzschild black hole with conformally coupled fields, for the semiclassical gravity equations cannot hold in the exterior region.

\subsection{Static vacuum in the static patch of de Sitter spacetime}
\label{AlphaVacua}

Let $(M, g_{ab})$ be the static patch of de Sitter spacetime with
\begin{align}
g_{ab} = -(1-H^2 r^2) \dd t_a \otimes \dd t_b + (1- H^2 r^2)^{-1} \dd r_a \otimes \dd r_b + r^2 \dd \mathbb{S}^2_{ab}, \quad H> 0.
\end{align}

The Riemann tensor, Ricci tensor and Ricci scalar are respectively
\begin{align}
R_{abcd} & = H^2 (g_{ac} g_{bc} - g_{ad} g_{bc}), \\ 
R_{ab} & = 3 H^2 g_{ab}, \\
R & = 12 H^2.
\end{align}

For fields with $m^2 + \xi R \leq 9 H^2/4$, de Sitter invariant vacuum states are defined by the Wightman function \cite{Schomblond:1976xc, Allen:1985ux}
\begin{align}
G^+_{\rm dS}(\x, \x') = \frac{8 \pi }{(m^2 + \xi R) - 2 H^2} \sec \left( \pi (9/4 - (m^2 + \xi R)/ H^2)^{1/2} \right) {}_2F_1(c, 3-c;2;(1+Z(\x, \x'))/2),
\end{align}
where ${}_2F_1$ is a hypergeometric function, the constant $c$ is a root of $c^2 -3 c + (m^2 + \xi R)/H^2$ %  $c := (1/2)\left[3- \left(9-4(m^2 + \xi R)/\alpha^2\right)^{1/2}\right]$ 
and where $Z(\x, \x') := \cos^2\left(2 H^2  \sigma_\epsilon(\x, \x') \right)$.

The expectation value of the renormalised stress-energy tensor is
\begin{align}
\omega_{\rm dS}(T_{ab}) & = \frac{96 \pi \xi}{(m^2/H^2 + 12 \xi )  - 2} \sec \left( \pi (9/4 - m^2/ H^2 + 12 \xi )^{1/2} \right) \left(\frac{2}{3 \Gamma (3-c) \Gamma (c)}+\frac{\psi(3-c)+\psi (c)+2 \gamma -1}{\Gamma (2-c) \Gamma (c-1)} \right)g_{ab} \nonumber \\
& - \frac{1}{2 (2 \pi)^2} \left( \frac{1}{8} m^4 + 3 \left(\xi - \frac{1}{6}\right)H^2 m^2   + \left[18 \left(\xi - \frac{1}{6} \right)^2 -\frac{1}{60} \right]H^4 \right) g_{ab},
\label{DeSitterStressEnergy}
\end{align}
where $\psi$ is the digamma function. Solutions to the semiclassical gravity equation are the roots of
\begin{align}
- 3 H^2 g_{ab} + \Lambda g_{ab} = 8 \pi G_N \omega_\alpha(T_{ab})
\end{align}
(setting $\omega_\alpha(T_{ab})$ as in eq. \eqref{DeSitterStressEnergy}) in the four-dimensional parameter space defined by the parameters $m^2, \xi, H, \Lambda$.

\subsection{Spacetimes with toroidal spatial section}
\label{subsec:Torus}

As mentioned above a Klein-Gordon field in the Minkowski vacuum in Minkowski spacetime yields a solution to semiclassical gravity if the renormalisation ambiguities are appropriately fixed. In this section we study the analogous problem for a locally flat spacetime with toroidal spatial section.

Let us consider spacetimes of the form $(\mathbb{R} \times \mathbb{T}^3, g_{ab})$, where $\mathbb{T}^3 = \mathbb{S}^1 \times \mathbb{S}^1 \times \mathbb{S}^1$ and $g_{ab} = - \dd t_a \otimes \dd t_b + \delta_{ij} \dd x^i_a \otimes \dd x^j_a$, with $t \in \mathbb{R}$ and each of the $x^i$ spatial coordinates ranging over $[0, 2 \pi L)$, i.e. the $\mathbb{S}^1$ have radius $L$. 

For time-periodic solutions the Klein-Gordon equation takes the form of an abstract wave equation, which reduces to the eigenvalue problem
\begin{equation}
A \psi_{n_i} = \omega_{n_i}^2 \psi_{n_i},
\end{equation}
where $A := - \triangle + m^2$ is positive and elliptic and $n_i \in \mathbb{Z}^3$. The eigenfunctions are of the form
\begin{align}
\psi_{n_i} (x) = \frac{1}{(2 \pi L)^{3/2} } \ee^{\ii L^{-1} n_i x^i }
\end{align}
and the eigenvalues are $\omega_{n_i}^2 = n^2/L^2 + m^2$. The initial data for the vacuum state is therefore \eqref{VaccumTwoPtSpectral}
\begin{align}
{}^\Omega G^+_{\varphi \varphi}(x, x') & = \lim_{\epsilon \to 0^+}\sum_{n_i \in \mathbb{Z}^3} \frac{\ee^{- \omega_{n_i} \epsilon}}{2 \omega_{n_i} } \psi_{n_i} (x) \overline{\psi_{n_i} (x')} \nonumber \\
& = \sum_{n_i \in \mathbb{Z}^3} \frac{m}{4 \pi^2} \frac{K_1 \left(m \left[(x^i-(x')^i+2\pi L n^i)(x_i-(x')_i+2\pi L n_i)+ \epsilon^2 \right]^{1/2} \right)}{\left[(x^i-(x')^i+2\pi L n^i)(x_i-(x')_i+2\pi L n_i) + \epsilon^2 \right]^{1/2} },
\label{ImagesTorus}
\end{align}
together with ${}^\Omega G^+_{\pi \pi} = A {}^\Omega G^+_{\varphi \varphi}$ and ${}^\Omega G^+_{\varphi \pi} = (\ii/2) 1\!\!1 = -{}^\Omega G^+_{\pi \varphi}$. The expression on the right-hand side of eq. \eqref{ImagesTorus} is particularly useful because the singular structure is contained in the $n_1 = n_2 = n_3 = 0$ term of the sum, while the remaining terms in the sum converge exponentially fast in the range of $x^i$ and $(x')^i$, see \cite{Hollands:2014eia}. % The expression agrees with the one expected from the method of images.

We are interested in obtaining the quantities $[w_\ell]$ and $[w_{\ell ;ij}]$ in order to analyse the constraints \eqref{Constraint3-DF} and \eqref{ConstraintTrAnom}, which read
\begin{subequations}
\begin{align}
\mathcal{C}_{ij} & = \frac{1}{2(2\pi)^2} \left[ -[w_{\ell;ij}]  - h_{ij} \frac{m^4}{8} \right] - \frac{\Lambda}{8 \pi G_{\rm N}} h_{ij} = 0, \label{CijTorus}\\
\mathcal{C}_{tt} & = \frac{1}{2(2 \pi)^2} \left[ m^2 [w_\ell] -  \frac{m^4}{4}  \right] + \frac{\Lambda}{2 \pi G_{\rm N}}  = 0. \label{CttTorus}
\end{align}
\end{subequations}

For concreteness, we choose the scale in the Hadamard subtraction piece as $\ell^2 = 2/m^2$. We have that
\begin{align}
[w_\ell]& = \frac{m^2}{16 \pi^2} \left( 2 \gamma - 1 \right) + \sum_{n_k \in \mathbb{Z}^3 \backslash \underline{0}} \frac{m}{8 \pi^3 L n} K_1 \left(2 \pi L m n \right), \\
% [w_{\ell;kl}] & =  \sum_{n_i \in \mathbb{Z}^3\backslash \underline{0}} \frac{n_k n_l \left[\left(2 \pi ^2 L^2 m^3 n^2+3m\right) K_1(2 \pi L m n  )+3 \pi  L m^2 n K_0(2 \pi L m n )\right]}{{16 \pi ^5 L^3 n^5}} -\delta_{kl}\left(1 - \frac{n_k n_l}{n^2} \right)\frac{ m^2  K_2(2 \pi L m n)}{16 \pi^4 L^2 n^2}, 
%[w_{\ell;kl}] & = \sum_{n_i \in \mathbb{Z}^3\backslash \underline{0}}  \delta_{k l}
% \frac{m^2}{16 \pi ^4 L^2 n^4} \left(2 \pi  L m n n_k n_l K_1(2 \pi L m n  )-\left(n^2-4 n_k n_l \right) K_2(2 \pi L m n )\right),
 [w_{\ell;ij}] & = \frac{ m \delta_{ij} }{(16 \pi)^2 (\pi L n)^5} \sum_{n_k \in \mathbb{Z}^3\backslash \underline{0}} \left[ K_1(2 \pi L m n) \left( (4 + 2 \pi^2 L^2 m^2 n^2) (2 \pi L)^2 n_i n_j - 2(4 \pi^2 L^2 n^2)  \right) \right. \nonumber \\
 & \left. - 2\pi L m n  K_0 (2 \pi L m n) \left( -2 (2 \pi L)^2 n_i n_j + 2 \pi^2 L^2 n^2 \right) \right]
\label{klTorus}
\end{align}
where $n = (n_i n^i)^{1/2}$. The term $[w_{\ell;kl}]$ in eq. \eqref{klTorus} has the form $[w_{\ell;kl}] = \breve w h_{kl}$, and hence constraint \eqref{CijTorus} can be written as a single equation,
\begin{align}
\frac{1}{2(2\pi)^2} \left[ -\breve w  - \frac{m^4}{8} \right] - \frac{\Lambda}{8 \pi G_{\rm N}}  = 0.
\label{CijDiagTorus}
\end{align}

Using constraints \eqref{CttTorus} and \eqref{CijDiagTorus} we have two equations and we should in principle be able to fix the ambiguities $\alpha_1$ and $\alpha_2$ in \eqref{Theta}, which appear implicitly in $\Lambda$ and $G_{\rm N}$ through \eqref{RenConstants}, to solve for these constraints. However, since $G_{ab} = 0$ the ambiguity $\alpha_2$ becomes redundant, in that only the combination $\Lambda/G_{\rm N}$ can be fixed with the ambiguities. In the case of the Minkowski vacuum in Minkowski spacetime, this is sufficient, since the vacuum state shares the symmetries of spacetime and the stress-energy tensor is proportional to the metric. However, in the case of a flat spacetime with toroidal spatial section the vacuum state is no longer Poincar\'e invariant, and as a result the $\mathcal{C}_{tt}$ and $\mathcal{C}_{ij}$ constraints are not proportional to each other.

 The $\mathcal{C}_{tt}$ constraint \eqref{CttTorus} imposes that
\begin{align}
 \frac{\Lambda}{8 \pi G_{\rm N}}  = \frac{m^2}{8(2 \pi)^2} \left( \frac{m^2}{4} - \frac{m^2}{16 \pi^2} \left( 2 \gamma - 1 \right) - \sum_{n_i \in \mathbb{Z}^3 \backslash \underline{0}} \frac{m}{8 \pi^3 L n} K_1 \left(2 \pi L m n \right)   \right),
 \label{CttTorusSolve}
\end{align}
and inserting eq. \eqref{klTorus}  and \eqref{CttTorusSolve} into \eqref{CijTorus} (or \eqref{CijDiagTorus}) we have
%\begin{align}
%  & -\sum_{n_i \in \mathbb{Z}^3\backslash \underline{0}}  \delta_{k l}
% \frac{m^2}{16 \pi ^4 L^2 n^4} \left(2 \pi  L m n n_k n_l K_1(2 \pi L m n  )-\left(n^2-4 n_k n_l \right) K_2(2 \pi L m n )\right) \nonumber \\
% &    = \delta_{kl} \frac{m^2}{4} \left[ \frac{3 m^2}{4}  - \frac{m^2}{16 \pi^2} \left( 2 \gamma - 1 \right) - \sum_{n_i \in %\mathbb{Z}^3 \backslash \underline{0}} \frac{m}{8 \pi^3 L n} K_1 \left(2 \pi L m n \right)   \right] 
%\end{align}
%\begin{align}
%  \sum_{n_i \in \mathbb{Z}^3 \backslash \underline{0}} & \left[\frac{m}{8 \pi^3 L n} K_1 \left(2 \pi L m n \right) + \frac{\delta_{kl}}{4 \pi ^4 L^2 n^4} \left(2 \pi  L m n n_k n_l K_1(2 \pi L m n  )-\left(n^2-4 n_k n_l \right) K_2(2 \pi L m n )\right) \right] \nonumber \\
%  &   =    \frac{3 m^2}{4}  - \frac{m^2}{16 \pi^2} \left( 2 \gamma - 1 \right).
%\label{ConstraintTorus}
%\end{align}
%
\begin{align}
% -   [w_{\ell;ij}] = \frac{m^2}{4} \left[ \frac{3 m^2}{4}  - \frac{m^2}{16 \pi^2} \left( 2 \gamma - 1 \right) - \sum_{n_i \in \mathbb{Z}^3 \backslash \underline{0}} \frac{m}{8 \pi^3 L n} K_1 \left(2 \pi L m n \right)   \right] h_{ij} 
%& \sum_{n_i \in \mathbb{Z}^3 \backslash \underline{0}} \left\{ \frac{m^3}{32 \pi^3 L n} K_1 \left(2 \pi L m n \right) -   \frac{ m  }{(16 \pi)^2 (\pi L n)^5}   \left[ K_1(2 \pi L m n) \left( (4 + 2 \pi^2 L^2 m^2 n^2) (2 \pi L)^2 \delta_{ij} n_i n_j - 2(4 \pi^2 L^2 n^2)  \right) \right. \right. \nonumber \\
% & \left. \left. - 2\pi L m n  K_0 (2 \pi L m n) \left( -2 (2 \pi L)^2 \delta_{ij} n_i n_j + 2 \pi^2 L^2 n^2 \right) \right] \right\} =  \left( \frac{3 m^4}{16}  - \frac{m^4}{64 \pi^2} \left( 2 \gamma - 1 \right) \right)
 & \sum_{n_i \in \mathbb{Z}^3 \backslash \underline{0}} \left\{ \frac{3 m^3}{32 \pi^3 L n} K_1 \left(2 \pi L m n \right) -   \frac{ m  }{(16 \pi)^2 (\pi L n)^5}   \left[   (-2 + 2 \pi^2 L^2 m^2 n^2) (2 \pi L)^2 n^2 K_1(2 \pi L m n)    \right. \right. \nonumber \\
 & \left.  + 4\pi^3 L^3 m n^3  K_0 (2 \pi L m n)  \right] \Big\} =   \frac{9 m^4}{16}  - \frac{3 m^4}{64 \pi^2} \left( 2 \gamma - 1 \right).
\label{ConstraintTorus}
\end{align}

Since the left-hand side of eq. \eqref{ConstraintTorus} is $L$-dependent and the right-hand side is $L$-independent, at a fixed renormalisation scale semiclassical gravity on the flat torus admits \emph{at most} $L$-fine-tuned solutions and generically does not admit solutions. Let us emphasise that in eq. \eqref{ConstraintTorus} there are no renormalisation ambiguities that can be fixed in order to obtain solutions for an arbitrary value of $L$. 

We should make one further remark here comparing the situation with Minkowski spacetime. As mentioned below eq. \eqref{CttTorus}, we have been working here at a fixed renormalisation scale, $\ell^2 = 2/m^2$. For the Minkowski vacuum, at any given renormalisation scale $\ell$ it is possible to find a solution to semiclassical gravity, and in the case of a locally flat spacetime with toroidal spatial section we have seen that at a fixed scale one cannot find solutions for arbitary $L$. However, for any given $L$ it is possible to find a scale at which the semiclassical gravity equations hold. Changing the scale modifies the constraint $\mathcal{C}_{ij} = 0$ \eqref{CijTorus} to $\mathcal{C}_{ij} = s m^4  h_{ij}$, where $s \in \mathbb{R}$. Accordingly, the constraint \eqref{CttTorus} changes to $\mathcal{C}_{tt} = -4 s m^4  h_{ij}$. Thus, one has a system of two constraints with two variables, the free ambiguities $\Lambda/G_{\rm N}$ and $s$, that can be solved and yield a solution to semiclassical gravity. 

%Eq. \eqref{klTorus} indicates that, for generic values of $L$ and $m^2$, $[\omega_{\ell;kl}]$ does not vanish for $k \neq l$, from where it follows that the constraint \eqref{CijTorus} cannot be satisfied. Hence, despite the vacuum state yielding a semiclassical gravity solution in Minkowski spacetime, the same need not be true for mere locally flat spacetimes.

%We consider here static spacetimes with metric
%\begin{align}
%g_{ab}(\x) = \Theta(\x)( - \dd t_a \otimes \dd t_b + \delta^{ij} \dd x^i_a \otimes \dd x^j_a ),
%\end{align}
%where $\Theta$ is time-independent.

% Nevertheless, as we shall see, local flatness does not suffice to generically yield solutions to semiclassical gravity. Instead, we will find that for a spacetime with toroidal spatial section there are \emph{at most} semiclassical gravity solutions in the vacuum state for finely-tuned torus radii, and generically no solutions are admitted irrespective of the values of $\Lambda$, $G_{\rm N}$, $\alpha$ and $\beta$. This highlights the important r\^ole played by the topology of spacetime.

\section{Summary and final remarks}
\label{sec:Conc}

We have studied semiclassical gravity in ultrastatic and static globally hyperbolic spacetimes. In the latter case, we have taken advantage of conformal techniques in order to detail the Hadamard subtraction step in the renormalisation of the stress-energy tensor, while in the former case the Hadamard bi-distribution is known explicitly due to the symmetry of spacetime.

Due to the presence of a global timelike Killing vector, the differential equations govening the evolution of geometry become merely constraints on the initial data of the Klein-Gordon field -- or more precisely on the regular part of the Wightman two-point function. Under a ``positive energy" assumption for the Klein-Gordon field, one can obtain a closed-form expression for the Wightman function in spacetime (i.e., as prescribing the correlations between fields at two non-null connected spacetime points) in terms of the Wightman function Cauchy data (which prescribes the correlations between fields and its momenta on a Cauchy surface). Moreover, using this closed-form expression, one can characterise initial data conditions that guarantee the stationarity of the state, and hence the preservation of the initial constraints at all times if they hold on the initial value surface.

The main results of the previous sections can be summarised as follows:

\begin{thm}
\label{ThmMain:Ultrastat}
Let $(M = \mathbb{R} \times S,g_{ab})$ be a globally hyperbolic, ultrastatic spacetime with Cauchy surface $S$, and with an irrotational, timelike Killing vector $\eta^a = \partial_t^a$  and metric $g_{ab} = - \dd t_a \otimes \dd t_b + h_{ij} \dd x^i_a \otimes \dd x^j_b$ in local coordinates $p^\mu = (t, x^i)$. The semiclassical Einstein field equations with a free Klein-Gordon field \eqref{semiEFE2} (or \eqref{semiEFE2-alt}) and \eqref{KG2} take the form
\begin{align}
(\partial_t^2 + A_h) G^+(p, p') = (\partial_{t'}^2 + A_{h'}) G^+(p, p') = 0, \quad \text{with} \quad A_h := -\Delta_h + m^2 + \xi \R
\label{MainThmUltraStat-KG}
\end{align}
where $\Delta_h$ is the Laplacian operator in the Riemannian spacetime $(S, h_{ab})$ with Riemann curvature $\mathcal{R}_{abcd}$, subject to constraints \eqref{semiEFE-constraints1} and \eqref{semiEFE-constraints3} preserved in time. The Hadamard condition on $G^+(p,p')$ is \eqref{HadamardFormUltrastatic}. Assume the ``positive energy" condition that $A_h$ is positive in $L^2(S, \dd \vol_h)$, then $A = \overline{A_h}$ is self-adjoint in $L^2(S, \dd \vol_h)$. 

Given initial data \eqref{WightmanIVP} at $t = 0$,
\begin{subequations}
\begin{align}
G^+((0,x),(0,x')) = \omega(\varphi(x) \varphi(x')) = G^+_{\varphi\varphi}(x,x'), & \quad & \partial_t G^+((0,x),(0,x')) = \omega(\pi(x) \varphi(x')) = G^+_{\pi \varphi}(x,x'), \\
\partial_{t'}G^+((0,x),(0,x')) = \omega(\varphi(x) \pi(x')) = G^+_{\varphi \pi}(x,x'), & \quad & \partial_t \partial_{t'} G^+((0,x),(0,x')) = \omega(\pi(x) \pi(x')) = G^+_{\pi \pi}(x, x'),
\end{align}
\end{subequations}
that satisfies the Hadamard condition \eqref{HadamardFormUltrastatic} on $S$, and as distributional kernels that
\begin{align}
G^+_{\pi \pi} = A G^+_{\varphi \varphi} \quad \text{ and } \quad G_{\varphi \pi}^+ = \frac{\ii}{2} 1\!\!1 = -G_{\pi \varphi}^+,
\end{align}
the solution to \eqref{MainThmUltraStat-KG} is as the integral kernel given of eq. \eqref{UltrastatInvWightman},
\begin{align}
G^+(t,t') & = G^+_{\varphi \varphi} \cos(A^{1/2} (t-t'))  -  \ii \frac{\sin (A^{1/2} (t-t')) }{2A^{1/2}}.
\label{ThmUltraStatInvWightman}
\end{align}

Eq. \eqref{ThmUltraStatInvWightman} defines time-translation invariant states. If the constraints \eqref{semiEFE-constraints1} and \eqref{semiEFE-constraints3} hold on $S$, then they hold everywhere in spacetime and eq. \eqref{ThmUltraStatInvWightman} is a solution to semiclassical gravity in ultrastatic spacetimes. Furthermore, using Synge's brackets notation for coincidence limits, constraints \eqref{semiEFE-constraints1} and \eqref{semiEFE-constraints3} take the form of eq. \eqref{Constraint3-DF} and \eqref{ConstraintTrAnom}.
\end{thm}

\begin{thm}
\label{ThmMain:Stat}
Let $(M, \tilde{g}_{ab})$ be a globally hyperbolic, static spacetime with metric $\tilde{g}_{ab} = \Theta^2 g_{ab} = \ee^{2 \theta} g_{ab}$, where $\Theta > 0$ is a smooth, time-independent, conformal factor and $g_{ab}$ is as in theorem \ref{ThmMain:Ultrastat}. The semiclassical Einstein field equations with a free Klein-Gordon field \eqref{semiEFE2} (or \eqref{semiEFE2-alt}) and \eqref{KG2} take the form
\begin{align}
(\partial_t^2 + \tilde{A}_h)\Theta \tilde{G}^+(p, p') = (\partial_{t'}^2 + \tilde{A}_{h'}) \Theta' \tilde{G}^+(p, p') = 0,
\label{MainThmStat-KG}
\end{align}
with $\tilde{A}_{h}$ an $L^2(S, \dd \vol_h)$ operator defined by eq. \eqref{Atilde}, subject to constraints \eqref{Ctt-Stat-Trace} (alternatively \eqref{CttStat}) and \eqref{CijStat} preserved in time. The Hadamard condition on $\tilde{G}^+(p,p')$ is \eqref{StaticHadamard}. Assume the ``positive energy" condition that $\tilde{A}_h$ is positive in $L^2(S, \dd \vol_h)$, then $\tilde{A} = \overline{\tilde{A}_h}$ is self-adjoint in $L^2(S, \dd \vol_h)$. Given initial data at $t = 0$
\begin{subequations}
\begin{align}
 \tilde{G}^+((0,x),(0,x')) = \Theta^{-1}(x) G^+_{\varphi \varphi}(x,x') \Theta^{-1}(x'), & \quad & \partial_t \tilde{G}^+((0,x),(0,x')) = \Theta^{-1}(x) G^+_{\pi \varphi}(x,x') \Theta^{-1}(x'), \\
\partial_{t'}\tilde{G}^+((0,x),(0,x')) = \Theta^{-1}(x) G^+_{\varphi \pi}(x,x') \Theta^{-1}(x'), & \quad & \partial_t \partial_{t'} \tilde{G}^+((0,x),(0,x'))  = \Theta^{-1}(x) G^+_{\pi \pi}(x, x') \Theta^{-1}(x'),
\end{align}
\end{subequations}
that satisfies the Hadamard condition \eqref{StaticHadamard} at $t = 0$, and as distributional kernels that
\begin{align}
G^+_{\pi \pi} = \tilde A G^+_{\varphi \varphi} \quad \text{ and } \quad G_{\varphi \pi}^+ = \frac{\ii}{2} 1\!\!1 = -G_{\pi \varphi}^+,
\end{align}
then
\begin{align}
\tilde G^+(t, t') & =  \Theta^{-1} \cos(\tilde A^{1/2}(t-t'))  G^+_{\varphi \varphi} (\Theta')^{-1} - \ii \Theta^{-1} \frac{\sin(\tilde{A}^{1/2}(t-t'))}{2 \tilde{A}^{1/2}} (\Theta')^{-1}
\label{MainThmSolStat}
\end{align}
is the integral kernel form of the bi-solution to eq. \eqref{MainThmStat-KG}. If the constraints \eqref{Ctt-Stat-Trace} (alternatively \eqref{CttStat}) and \eqref{CijStat} hold for the initial data, then they hold everywhere in spacetime and \eqref{MainThmSolStat} is a solution to the semiclassical Einstein field equations in static spacetimes.
\end{thm}

Admittedly, given the form of the initial-data constraints, even the existence of the vacuum state in a given spacetime is not guaranteed. Nevertheless, we have aimed at expressing the constraints in a way that allows for an efficient check, given a state of interest, cf. eq. \eqref{Constraint3-DF} and \eqref{ConstraintTrAnom} in the ultrastatic case and eq. \eqref{CijStat} and \eqref{Ctt-Stat-Trace} in the static case. %Moreover, the form of the constraint equations also allows one to \emph{generate ``good" initial data} for the Wightman function, given a fixed initial 3-geometry. 
One of the constraint equations involves an elliptic equation that fixes the diagonal of the regular part of the two-point function, cf. eq. \eqref{ConstraintTrAnom} in the ultrastatic case and eq. \eqref{Ctt-Stat-Trace} in the static case. Once this constraint is solved, the result can be inserted into eq. \eqref{Constraint3-DF} in the ultrastatic and eq. \eqref{CijStat} in the static case and there remains a set of three constraints to be solved. %This opens the possibility of fixing the diagonal of the regular part of the state and reconstruct the initial value correlations -- at least in a convex normal neighbourhood -- using the Hadamard recursion relations.
Despite the complicated form of the constraints, it is possible to deduce that for conformally-coupled fields in Ricci-flat spacetimes semiclassical gravity only admits solutions in Minkowski spacetime, see prop. \ref{Prop-no-go}. Some examples are discussed in sec. \ref{sec:Examples}, including for ultrastatic spacetimes with hyperbolic spatial section, de Sitter spacetime and the flat spacetime with toroidal spatial section.

It seems to us that the techniques that we have studied here yield a promising avenue for checking ``good" initial data for the Wightman two-point function in full semiclassical gravity, in situations in which one has an asymptotically static region. For spacetimes that are not asymptotically static, one could perhaps still interpolate between a static spacetime and the spacetime of interest with the aid of auxiliary switching functions, then ``forget" about the static region, and in this way generate ``good" initial data for the Wightman function in the spacetime of interest, where the form of the Hadamard bi-solution is a priori not known, as argued in sec. IV B of \cite{Juarez-Aubry:2019jon} in the context of a semiclassical scalars toy model. The key point is that, indeed, if the state is Hadamard in the static region (containing a Cauchy surface), then it must be Hadamard in all of spacetime \cite{Fulling:1981cf}. These ideas, however, deserve careful further examination. 

As a final remark, we should point out that the conformal techniques that we have used for studying static spacetimes can be well adapted to studying semiclassical cosmology, even in the presence of anisotropies, i.e., with a spacetime-dependent scale factor. This is a more complicated case since now the semiclassical Einstein equations provide the evolution law for the scale factor, and are not merely constraints. Advances in this direction for semiclassical gravity in globally hyperbolic spacetimes with conformally covariant fields appears in \cite{Juarez-Aubry:2021abq}.

\section*{Acknowledgments}
This work is funded by a CONACYT Postdoctoral Research Fellowship. We also thank the support of CONACYT project no. 140630 and UNAM-DGAPA-PAPIIT grant no. IG100120. The author gladly thanks Bernard S. Kay, Tonatiuh Miramontes and Daniel Sudarsky for many stimulating conversations on semiclassical gravity. We thank two anonymous referees for comments on an earlier version of this work.

\appendix

\section{An application of formula \eqref{deltaFormula}}
\label{Appendix}

\begin{ex}
We present an application of formula \eqref{deltaFormula}. Consider a test function $f \in C_0^\infty(\mathbb{R}^3)$. We have
\begin{align}
& \lim_{\epsilon \to 0^+}- \int_{\mathbb{R}^3} \dd \vol f \frac{\epsilon}{2(2\pi)^2} \left( -\frac{|\Delta_h|^{1/2}}{(\sigma_h + \frac{1}{2} \epsilon^2)^2} + \frac{v}{\sigma_h + \frac{1}{2}\epsilon^2}   \right) \nonumber \\ 
&  = \lim_{\epsilon \to 0^+} -\frac{  \epsilon}{2(2 \pi)^2} \int_{\mathbb{R}^3} \dd r \dd \theta \dd \phi r^2 \sin \theta f(r,\theta, \phi) \left( \frac{-1}{(r^2/2 + \epsilon^2/2)^2} + \frac{v}{r^2/2 + \epsilon^2/2} \right) \nonumber \\
& =  \lim_{\epsilon \to 0^+} -\frac{ 1 }{2(2 \pi)^2} \int_{\mathbb{R}^3}  \dd x \dd \theta \dd \phi x^2 \sin \theta f(\epsilon x, \theta, \phi) \left( \frac{-4}{(x^2 + 1 )^2} + \frac{2 v \epsilon^2}{x^2 + 1 } \right),
\end{align}
where the last line is obtained by a change of variables. Using the dominated convergence theorem
\begin{align}
 \lim_{\epsilon \to 0^+}- \int_{\mathbb{R}^3} \dd \vol f \frac{\epsilon}{2(2\pi)^2} \left( -\frac{|\Delta_h|^{1/2}}{(\sigma_h + \frac{1}{2} \epsilon^2)^2} + \frac{v}{\sigma_h + \frac{1}{2}\epsilon^2}   \right)   & =  \frac{1 }{2 \pi^2} f(0)\int_{\mathbb{R}^3}  \dd x \dd \theta \dd \phi x^2 \sin \theta  \frac{1}{(x^2 + 1 )^2} \nonumber \\
& = \frac{2  }{\pi} f(0)\int_{\mathbb{R}}  \dd x    \frac{x^2}{(x^2 + 1 )^2} =  f(0),
\end{align}
as desired.
\end{ex}

\begin{rem}
Formula \eqref{deltaFormula} can be generalised to be used even if the support of $f$ is not contained inside the convex normal neighbourhood containing the points $x$ and $x'$. In this case, one can use a characteristic function, $\chi$, such that ${\rm supp} (\chi f)$ is contained in such convex normal neighbourhood and apply formula \eqref{deltaFormula}.  
\end{rem}

\end{document}